\documentclass[11pt, letterpaper]{article}

\usepackage{fullpage}
\usepackage{amsthm}
\usepackage{amsmath,amssymb,amsfonts,nicefrac}
\usepackage{xspace}
\usepackage{color}
\usepackage{url}
\usepackage{bm}
\usepackage{bbm}
\usepackage{times}
\newtheorem{theorem}{Theorem}[section]
\newtheorem*{thm2}{Theorem}

\newtheorem{definition}[theorem]{Definition}
\newtheorem{proposition}[theorem]{Proposition}
\newtheorem{conjecture}[theorem]{Conjecture}
\newtheorem{lemma}[theorem]{Lemma}

\newtheorem{corollary}[theorem]{Corollary}
\newcommand{\junk}[1]{}
\newcommand{\ignore}[1]{}
\newcommand \E {{\rm E}}
\newcommand{\D}[0]{{\ensuremath{\mathcal{D}}}\xspace}

\newcommand{\R}[0]{{\ensuremath{\mathbb{R}}}}

\newcommand{\Ft}{\ensuremath{\mathbb{F}[2]}}

\newcommand{\eps}{\varepsilon}

\newcommand{\ol}[1]{\ensuremath{\overline{#1}}\xspace}
\newcommand{\wh}[1]{\ensuremath{\widehat{#1}}\xspace}
\newcommand{\mc}[1]{\ensuremath{\mathcal{#1}}\xspace}
\newcommand{\tn}[1]{\ensuremath{\textnormal{#1}}\xspace}
\newcommand{\poly}{\ensuremath{\textnormal{poly}}\xspace}

\newcommand{\IS}{{\sc AlmostColHyp}}
\newcommand{\wt}{\textnormal{wt}}
\newcommand{\slfrac}[2]{\left.#1\middle/#2\right.}

\newcommand{\initOneLiners}{%
    \setlength{\itemsep}{0pt}
    \setlength{\parsep }{0pt}
    \setlength{\topsep }{0pt}
}

\def\showauthornotes{1}
\ifnum\showauthornotes=1
\newcommand{\Authornote}[2]{{\sf\small\color{red}{[#1: #2]}}}
\else
\newcommand{\Authornote}[2]{}
\fi

\newcommand{\defeq}{\stackrel{\textup{def}}{=}}

\begin{document}
\title{Hardness of Finding Independent Sets in $2$-Colorable \\ 
and Almost
$2$-Colorable Hypergraphs}
\author{Subhash Khot\thanks{Department of Computer Science, 
University of Chicago, USA. email: \texttt{khot@cs.nyu.edu}}
\and
Rishi Saket\thanks{IBM T. J. Watson Research Center, USA. email: \texttt{rsaket@us.ibm.com} }}
\maketitle
\thispagestyle{empty}
\setcounter{page}{0}
\begin{abstract}
This work studies the hardness of finding independent sets in
hypergraphs which are either $2$-colorable or are \emph{almost}
$2$-colorable, i.e. can be $2$-colored after removing a small fraction
of vertices and the incident hyperedges. To be precise, say that a
hypergraph is $(1-\eps)$-almost $2$-colorable if removing an $\eps$
fraction of its vertices and all hyperedges incident on them makes the
remaining hypergraph $2$-colorable.
In particular we prove the
following results.
\begin{itemize}
\item For an arbitrarily
small constant $\gamma > 0$, there is a constant $\xi > 0$, 
such that, given a $4$-uniform hypergraph on $n$
vertices which is $(1 - \eps)$-almost $2$-colorable for $\eps =
2^{-(\log n)^\xi}$, it is quasi-NP-hard\footnote{A problem is
quasi-NP-hard if it admits a $n^{\poly(\log n)}$ time reduction from
$3$SAT.} to find an independent set of
$\slfrac{n}{\left(2^{(\log
n)^{1-\gamma}}\right)}$ vertices.   

\item 
For any constants $\eps, \delta > 0$, given as input a 
$3$-uniform hypergraph on $n$ vertices which is $(1-\eps)$-almost 
$2$-colorable, 
it is NP-hard to find an
independent set of $\delta n$ vertices.

\item  
Assuming the 
\emph{$d$-to-$1$ Games Conjecture} the following holds.
For any constant $\delta > 0$, 
given a $2$-colorable $3$-uniform hypergraph
on $n$ vertices, it is NP-hard to find an
independent set of $\delta n$ vertices.
\end{itemize}
The hardness result on independent set in almost $2$-colorable
$3$-uniform hypergraphs was earlier known only assuming the Unique Games
Conjecture. In this work we prove the result \emph{unconditionally},
combining Fourier analytic techniques with the Multi-Layered PCP of
\cite{DGKR03}.

For independent sets in $2$-colorable $3$-uniform hypergaphs we prove the 
first strong hardness result, albeit assuming the $d$-to-$1$ Games
Conjecture. Our reduction uses the $d$-to-$1$ Game as a starting point
to construct a
Multi-Layered PCP with the \emph{smoothness} property. We use
analytical techniques based on the Invariance
Principle of Mossel~\cite{Mossel}. The smoothness property is
crucially exploited in a manner similar to recent work of H\aa
stad~\cite{Hastad12} and Wenner~\cite{Wenner12}.    

Our result on almost $2$-colorable $4$-uniform hypergraphs 
gives the first nearly polynomial hardness factor
for independent set in hypergraphs which are (almost) colorable with
constantly many colors. It partially bridges the gap between the
previous best lower bound of $\poly(\log n)$ and the algorithmic upper
bounds of $n^{\Omega(1)}$. This also exhibits a bottleneck to
improving the algorithmic techniques for hypergraph coloring.

\end{abstract}

\section{Introduction}
A $k$-uniform hypergraph consists of a set of vertices and a set of
hyperedges, where each hyperedge is a subset of exactly $k$
vertices. For $k=2$ this defines the usual notion of a graph. An
\emph{independent set} in a $k$-uniform hypergraph is a subset of vertices
such that no hyperedge has all of its $k$ vertices from this subset.
In other words, an independent set does not contain any hyperedge. The
problem of finding independent sets of maximum size in (hyper)graphs is
a fundamental one in combinatorial optimization. Note that the
complement of an independent set is a \emph{vertex cover}, i.e. a
subset of vertices that contains at least one vertex from each
hyperedge. Thus, finding a maximum sized independent
set is  same as finding a minimum vertex cover, an equally
important problem in combinatorics. Throughout this paper, we shall
frequently
use the size of a set of vertices to mean its \emph{relative} size,
i.e. as a fraction of the total weight of the vertices.

The study of independent sets is closely
related to that of hypergraph coloring. A hypergraph is 
$q$-colorable if its vertices can each be assigned one of $q$ distinct
colors so that no hyperedge is monochromatic. The
problem in hypergraph coloring is to determine the
minimum possible value of $q$, which is known as the \emph{chromatic
number} of the hypergraph. Note
that the color classes in a $q$-coloring form a partition of the
vertices into $q$ disjoint independent sets. Thus, a $q$-colorable
hypergraph has an independent set of size at least $\slfrac{1}{q}$. 
On the other hand, if a hypergraph does not have an
independent set of size $\slfrac{1}{q}$  then it is not
$q$-colorable either. Thus, the absence of large independent sets
implies a large chromatic number. 

This connection can also be studied with a
relaxed notion of hypergraph coloring. Say that a hypergraph is \emph{
almost $q$-colorable} if there is a subset of vertices of size at most
 $\eps$ such that removing this subset and all
hyperedges containing a vertex from this subset makes the hypergraph
$q$-colorable. Here $\eps$ can be an arbitrarily
small positive constant. It is easy to see that an almost
$q$-colorable hypergraph contains $q$ pairwise disjoint independent sets 
containing within them at least $(1-\eps)$ fraction of vertices. Thus,
there is at least one independent set of size $\slfrac{(1-\eps)}{q}$. 

The problem of finding independent sets
in (almost) $q$-colorable $k$-uniform hypergraphs is most interesting for small values of $q$
and $k$ and has been studied extensively from the complexity
perspective in a sequence of works including
\cite{GHS, Khot-color,Holmerin,Khot-3,DRS,BK09,BK10,GSinop,KS12,Chan13}. 
For constant $q$ and $k$, the strongest hardness 
result in terms of the relative size of the independent set is by
Khot~\cite{Khot-color} who showed the hardness of finding independent sets of
size $(\log n)^{-c}$ in $5$-colorable $4$-uniform
hypergraphs on $n$ vertices. On the other hand, the best algorithms for
these problems yield independent sets of size $n^{-\Omega(1)}$. 
 
In this work we focus on the case of (almost)
$2$-colorable $3$-uniform and $4$-uniform hypergraphs. The motivation
for our first result stems from the gap between the algorithmic and
complexity results mentioned above. We prove the following.
\begin{theorem}\label{thm-main1}
For any arbitrarily
small constant $\gamma > 0$, there is a constant $\xi > 0$
such that given a $4$-uniform hypergraph $G(V, E)$ on $n$
vertices such that removing $2^{-(\log n)^\xi}$ fraction of vertices and
all hyperedges incident on them makes the remaining hypergraph
$2$-colorable, it is quasi-NP-hard to find an independent set in $G$ of
$\slfrac{n}{\left(2^{(\log
n)^{1-\gamma}}\right)}$ vertices.    
\end{theorem}
This is the first result showing an almost polynomial factor hardness
for independent set in (almost) $q$-colorable $k$-uniform hypergraphs.
While existing algorithms are for the case of exact colorability, they
rely on the presence of a small number of 
pairwise disjoint independent sets 
covering almost
all the vertices,
and are also applicable to the case of almost colorability. Thus,
the above result indicates a bottleneck in the improvement of existing
algorithms. The hardness factor obtained is exponentially stronger than
the previous lower bound of $\poly(\log n)$ by Khot~\cite{Khot-color},
albeit for the case of exact colorability.

Our next result is an analogue of the result of Bansal and
Khot~\cite{BK09,BK10} who showed, assuming the Unique Games Conjecture
(UGC),
that it is NP-hard to find an independent set of $\delta$ fraction
of vertices (for any constant $\delta > 0$) in an almost $2$-colorable
graph (i.e. almost bipartite graph). The related work of Guruswami and
Sinop~\cite{GSinop} showed a similar result for almost $2$-colorable
$3$-uniform hypergraphs (with the hardness factor depending on the
degree), assuming UGC. We show that it is possible to
prove the result for $3$-uniform hypergraphs \emph{without}
assuming UGC. 
\begin{theorem}\label{thm-main2}
For any constants $\eps, \delta > 0$, given a $3$-uniform hypergraph
on $n$ vertices
such that removing at most $\eps$ fraction of vertices and the
hyperedges incident on them makes the remaining hypergraph $2$-colorable, 
it is NP-hard to find an
independent set of $\delta n$ vertices. 
\end{theorem}
The instances constructed in the Theorems \ref{thm-main1} and
\ref{thm-main2} are degree regular, and thus also work for an
alternate definition of almost colorability -- which involves removing
$\eps$ fraction of the hyperedges instead of vertices --  used in
\cite{GSinop}.

Our final result proves the first strong hardness factor for finding
independent sets in $2$-colorable $3$-uniform hypergraphs, assuming
 the
\emph{$d$-to-$1$ Games Conjecture} of Khot~\cite{Khot02}. 
\begin{theorem}\label{thm-main3} 
Assuming the 
\emph{$d$-to-$1$ Games Conjecture} the following holds.
For any constant $\delta > 0$,
given a $2$-colorable $3$-uniform hypergraph
on $n$ vertices, it is NP-hard to find an
independent set of $\delta n$ vertices.
\end{theorem}
We note that Dinur, Regev and Smyth~\cite{DRS} showed
that $2$-colorable $3$-uniform hypergraphs are NP-hard to color with
constantly many colors. However, their reduction produced instances
with linear sized independent sets in the NO Case, and thus did not
yield any hardness for finding independent sets in such hypergraphs. 
Our result therefore proves a stronger property, albeit
assuming the conjecture.

In the remainder of this section we shall formally state the problems
we study in this work, give an overview of previous related work and
describe the techniques used in our results.

\subsection{Problem Definition}
Given a hypergraph $G$, let ${\sf IS}(G)$ be the size of the maximum
independent set in $G$ and let $\chi(G)$ be its chromatic number, i.e. the
minimum number of colors required to color the hypergraph such that every
hyperedge is non-monochromatic. 
We define the problem of finding independent sets
in $q$-colorable hypergraphs as follows.

\medskip
\noindent
{\bf {\sc ISColor}$(k, q, Q)$} : 
Given a $k$-uniform hypergraph $G(V, E)$, decide between,
\begin{itemize}
\item YES Case: $\chi(G) \leq q$.
\item NO Case: ${\sf IS}(G) < \frac{|V|}{Q}$.
\end{itemize}
It is easy to see that if {\sc ISColor}$(k, q, Q)$ is NP-hard for some
parameters $q, Q \in \mathbb{Z}^+$ then it is
NP-hard to color a $q$-colorable $k$-uniform hypergraph with $Q$ colors. 
In this paper
we also study a slight variant of this problem, in which the goal is
to find independent sets in almost colorable hypergraphs. For parameters
$k, q, Q$, and a parameter $\eps >0$ it is defined as follows.

\medskip
\noindent
{\bf {\sc ISAlmostColor}$_\eps(k, q, Q)$}: Given a $k$-uniform hypergraph 
$G(V, E)$, decide
between,
\begin{itemize}
\item YES Case: There is a subset of $(1-\eps)$ fraction of the
vertices, such that for 
the $k$-uniform hypergraph $G'$ on this subset of vertices 
containing the hyperedges which lie completely inside it, $\chi(G')
\leq q$. We also denote this by $\chi_\eps(G) \leq q$.
\item NO Case: ${\sf IS}(G) < \frac{|V|}{Q}$.
\end{itemize}
Note that the second property above, i.e. ${\sf IS}(G) < \frac{|V|}{Q}$, 
implies that $\chi_\eps(G) \geq Q-1$ for
sufficiently small $\eps > 0$.

Using the above definitions the results of this paper can be concisely
restated as follows. The number of vertices in the hypergraph is
denoted by $n$.

\subsubsection*{Our Results}

\begin{thm2}(Theorem \ref{thm-main1})
For an arbitrarily
small constant $\gamma > 0$, there is a constant $\xi > 0$ such that
{\sc ISAlmostColor}$_\eps(4, 2, Q)$ is quasi-NP-hard, where $\eps = 2^{-(\log
n)^\xi}$ and $Q = 2^{(\log n)^{1-\gamma}}$.
\end{thm2}

\begin{thm2} (Theorem \ref{thm-main2}) 
For any constant $Q > 0$ and arbitrarily small constant $\eps > 0$,\\ 
{\sc ISAlmostColor}$_\eps(3, 2, Q)$ is NP-hard. 
\end{thm2}

\begin{thm2} (Theorem \ref{thm-main3}) Assuming the $d$-to-$1$
Games Conjecture the following holds.
For any constant $Q > 0$, {\sc ISColor}$(3, 2, Q)$ is NP-hard. 
\end{thm2}

\subsection{Previous Work}
The problem of finding independent sets in (almost) colorable graphs
and hypergraphs has been studied extensively from algorithmic as
well as complexity perspectives. On $2$-colorable, i.e. bipartite
graphs, the maximum independent set can be computed in polynomial
time. On the other hand, a significant body of work -- including 
\cite{Wigderson}, \cite{Blum}, \cite{KMS}, \cite{BK},
\cite{ACC}, and \cite{KT12} -- 
has shown that a $3$-colorable graph can be efficiently
colored with $n^\alpha$ colors where the currently best value of
$\alpha \approx 0.2038$ was shown in \cite{KT12}. In particular, this
shows that {\sc ISColor}$(2, 3, n^{\alpha})$ can be efficiently
solved. For $2$-colorable
$3$-uniform hypergraphs Krivelevich et al.~\cite{KNS} gave a coloring
algorithm using $O(n^{1/5})$ colors, thus solving {\sc ISColor}$(3, 2,
O(n^{1/5}))$. Chen and Frize~\cite{CF} and Kelsen, Mahajan and
Ramesh~\cite{KMR} independently gave algorithms for coloring
$2$-colorable $4$-uniform hypergraphs using $O(n^{3/4})$ colors, which
implies an algorithm for  {\sc ISColor}$(4, 2,
O(n^{3/4}))$. In related work Chlamtac and Singh~\cite{CS} gave an
algorithm that on 
a $3$-uniform hypergraph which has an independent set of $\gamma
n$ vertices, efficiently computes an independent set of 
$n^{\Omega(\gamma^2)}$ vertices. While the algorithmic approaches have
studied the case of exactly colorable hypergraphs, 
they rely on the existance of
disjoint independent sets and are also applicable to almost colorable
hypergraphs.   

Several hardness results for these problems have been obtained
using either the PCP Theorem or well known conjectures 
as the starting point. 
Under standard  complexity assumptions,
Khot \cite{Khot-color} showed the hardness of finding independent sets of
size $(\log n)^{-c}$ in $5$-colorable $4$-uniform
hypergraphs on $n$ vertices. Building upon similar work of Guruswami,
H\aa stad and Sudan~\cite{GHS}, 
Holmerin~\cite{Holmerin} showed that it is NP-hard to
find an independent set of size $\delta$ in a
$2$-colorable $4$-uniform hypergraph, where $\delta > 0$ is any
constant. For $3$-uniform hypergraphs which
are $3$-colorable, Khot~\cite{Khot-3} showed a hardness of finding
independent sets of size $(\log\log n)^{-c}$.  On
$3$-colorable
graphs, assuming the so called \emph{Alpha Conjecture}, Dinur
et al.~\cite{DMR} showed it is NP-hard to find independent sets of
size $\delta$. Bansal and Khot~\cite{BK09, BK10}
assumed the more well known Unique Games Conjecture to show that it is
NP-hard to find independent sets of size $\delta$ in
\emph{almost bipartite} (i.e. almost $2$-colorable) graphs. Guruswami
and Sinop~\cite{GSinop} showed a similar result for almost $2$-colorable
$3$-uniform hypergraphs, the focus of their work being the case of
bounded degree hypergraphs.  

\medskip
It is pertinent to note that while the algorithmic results have
$\poly(n)$ factors, the previous best inapproximability was a
$\poly(\log n)$ factor~\cite{Khot-color}. Our result for independent
set in almost $2$-colorable $4$-uniform hypergraphs -- Theorem
\ref{thm-main1} -- partially bridges this gap by showing an almost
polynomial factor $2^{-(\log n)^{1-\eps}}$, an exponential improvement
over the previous lower bound.

Theorem \ref{thm-main2} unconditionally proves the hardness result for
independent set in almost $2$-colorable $3$-uniform hypergraphs,
which was earlier known only assuming the Unique Games Conjecture. We
also show -- in Theorem \ref{thm-main3} -- the first inapproximability
for the case of $2$-colorable $3$-uniform hypergraphs assuming
 the $d$-to-$1$ Games Conjecture.

In the rest of this section we give an informal overview of the
techniques used to proves our results.

\subsection{Overview of Techniques}
The results of this work follow a common template of
reductions from an instance of a NP-hard constraint satisfaction
problem -- the so called \emph{Outer Verifier} -- via its combination 
with a proof encoding -- the \emph{Inner Verifier}. However, the techniques
used to prove Theorems \ref{thm-main1}, \ref{thm-main2} and
\ref{thm-main3} are somewhat varied and we describe them separately. 

\subsubsection*{Almost $2$-Colorable $4$-Uniform Hypergraphs}
The goal of this result is to prove an almost polynomial hardness
factor for independent set in almost $2$-colorable $4$-uniform
hypergraphs. To accomplish this, the size of the
hardness reduction needs to be bounded. Thus, one cannot use \emph{Long
Codes} which have an unmanageable blowup for our purpose. Instead, we
use Hadamard Codes which are exponentially shorter and have
been used in previous works~\cite{KP06, KS08} for a similar reason.  
The Hadamard Code
$H^v$ of an element $v \in \Ft^m$ is indexed by all $x \in \Ft^m$ such
that $H^v(x) := x\cdot v \in \Ft$. The ``gadget'' used for the
reduction is as follows. 

Consider the following $4$-uniform
hypergraph. The vertex set is $\Ft^m$. Let $e_1 \in \Ft^m$ be the
element which has $1$ in the first coordinate and $0$ everywhere else.
For any $x, y, z \in \Ft^m$, add a hyperedge  between the elements 
$x, y, x+z$ and $y+z+e_1$, where the addition is done in the vector
space $\Ft^m$. This is (essentially) a $4$-uniform hypergraph.
Consider any element $v \in \Ft^m$ such that $v_1 = 1$. It is easy to
see that $H^v(x) + H^v(x+z) + H^v(y) + H^v(y + z + e_1) = 1$, and thus
the coloring to $\Ft^m$ given by the value of $H^v$ is a valid
$2$-coloring of this hypergraph.  On the other hand it can be shown
that any independent set $S \subseteq \Ft^m$ of size $\delta 2^m$ can  be
\emph{decoded} into a list of elements $v$ such that $v_1 = 1$. This
analysis uses only some basic tools from Fourier Analysis.

The above gadget can be combined with a parallel repetition of an
appropriate linear constraint system. In our case, we choose a
specialized instance of {\sc Max-$3$Lin} constructed by Khot and
Ponnuswami~\cite{KP06}. The main idea in this combination is to do the 
\emph{folding} only over the homogeneous constraints and use the
non-homogeneous constraints to play the role of $e_1$ in the above
gadget. The almost polynomial hardness factor is obtained by an
appropriate number of rounds of parallel repetition which is afforded
by the parameters of the {\sc Max-$3$Lin} instance used in the
reduction.

\subsection*{Almost $2$-Colorable $3$-Uniform Hypergraphs}
This reduction uses as the Outer Verifier a layered constraint
satisfaction problem, referred to as the \emph{Multi-Layered PCP}.
This PCP was used earlier by Khot~\cite{Khot-3} for similar results
for $3$-Colorable
$3$-Uniform Hypergraphs and by Dinur, Guruswami, Khot and 
Regev~\cite{DGKR03}
and Sachdeva and Saket~\cite{SS11} in
their hardness results for hypergraph vertex cover. Due to some
fundamental
limitations of existing techniques, the use of this PCP is
necessitated for proving results for independent sets in $3$-uniform 
hypergraphs.     

The Inner Verifier uses a \emph{biased} Long Code encoding similar to
the reductions of Dinur, Khot, Perkins and Safra~\cite{DKPS}, Khot 
and Saket~\cite{KS12}
and Sachdeva and Saket~\cite{SS13}. The following gadget encapsulates
the Inner Verifier. Consider the biased Long Code $\mc{H} =
\{1,2,*\}^m$. The associated measure is induced by sampling each
coordinate independently to be $1$ or $2$ with probability
$\frac{1-\eps}{2}$ and $*$ with probability $\eps$. Let $\mc{H}_0,
\mc{H}_1, \dots, \mc{H}_d$ be $d+1$ identical copies of $\mc{H}$. A vertex
weighted $3$-uniform hypergraph is constructed by taking the union of the $d+1$
Long Codes with weights given by the measure. Consider $x \in
\mc{H}_0$ and $y, z \in \mc{H}_k$ ($1\leq k \leq d$), such that for
any $i \in [m]$ the tuple $(x_i, y_i, z_i)$ is not $(1,1,1)$ or
$(2,2,2)$. Add a hyperedge between $x, y$ and $z$ for all such
choices. It is easy to see that for any $j \in [m]$, removing all the 
vertices $x$ such that $x_j = *$ and all hyperedges incident on these
vertices makes the hypergraph $2$-colorable by coloring the rest of
the vertices $y$ according to whether $y_j = 1$ or $2$. On the other
hand, using Russo's Lemma and Friedgut's Junta Theorem one can show
that if there is an independent set $\mc{I}$ which has at least 
$\delta$ fraction of measure from each of the $d+1$ Long Codes, then
it can be decoded into a distinguished coordinate $\ell \in [m]$. This
Inner Verifier is robust enough to be combined with the Multi-Layered
PCP to yield the desired result. 

The hardness factor obtained, however, is much weaker than in the
previous reduction, due to our use of Long Codes and also due to the
structure of the Multi-Layered PCP.

\subsubsection{$2$-Colorable $3$-Uniform Hypergraphs}
For independent set in $2$-colorable $3$-uniform hypergraphs, the
existing PCP techniques seem insufficient to yield the desired
results. Thus, we rely on the $d$-to-$1$ Games Conjecture of
Khot~\cite{Khot02}. This conjecture was earlier used to
establish hardness results for independent sets in $4$-colorable
graphs~\cite{DMR}. Our use of this conjecture is similar to that of
O'Donnell and Wu~\cite{OW} who showed an optimal $\frac{5}{8} + \eps$
factor hardness for a satisfiable instance of {\sc Max-$3$CSP}. In
a recent work H\aa stad~\cite{Hastad12} showed the same result
unconditionally. We also make use of certain techniques used in
\cite{Hastad12}. 

The Outer Verifier in our reduction is a multi-layered PCP constructed
using the $d$-to-$1$ games problem. The construction of this PCP
ensures a \emph{smoothness} property which has been used in several
previous works~\cite{Khot-3, KS06, KS08a, GRSW} including the above
mentioned work of H\aa stad~\cite{Hastad12} and a related work of
Wenner~\cite{Wenner12}. The Inner Verifier yields a $3$-uniform
hypergraph with hyperedges corresponding to a $3$-query PCP test over
Long Codes which
is in a same vein as the test used in \cite{OW} and \cite{Hastad12}. 
The analysis is based in large part on the Invariance Principle 
of Mossel~\cite{Mossel}, the application of which follows an approach
used by O'Donnell and Wu~\cite{OW}, while avoiding certain
complications they face. The smoothness property is crucial for the
analysis and is leveraged in a manner similar to \cite{Hastad12}.

\medskip
\noindent
{\bf Organization of Paper.} 
The next section contains the known PCP constructions which shall be
the starting points in our reductions for Theorems \ref{thm-main1} and
\ref{thm-main2}. We shall also state the $d$-to-$1$
Games Conjecture that we shall require for proving Theorem
\ref{thm-main3} and describe the smooth layered PCP we construct
based on this assumption, a sketch of the construction being 
deferred
to Section \ref{sec-dto1multi}.

Sections \ref{sec-main1}, \ref{sec-main2} and \ref{sec-main3} contain
the  hardness reduction and proofs for Theorems 
\ref{thm-main1}, \ref{thm-main2} and \ref{thm-main3} respectively
along with a description of the mathematical tools needed to complete
the analyses.
 
\section{Preliminaries}
In this section we shall describe some useful results in PCPs and
hardness of approximation along with the description of the 
 $d$-to-$1$ Games Conjecture. 

For proving Theorem \ref{thm-main1} we shall begin with the following 
theorem of Khot and
Ponnuswami~\cite{KP06} on the hardness of a specific {\it gap} version
 of {\sc Max-$3$Lin} with a desirable setting of the parameters. 
An instance of {\sc Max-$3$Lin} consists of a system of linear
equations over $\Ft$ where each equation has exactly $3$ variables,
the goal being to find an assignment to the variables 
satisfying the maximum number of equations. The instance is said to be
$d$-regular if each variable occurs in exactly $d$ equations.  
\begin{theorem}
\label{thm-KP} \cite{KP06} Given a $7$-regular instance $\mathcal{A}$
of  {\sc Max-$3$Lin} over $\Ft$
on $n$ variables, unless $\textnormal{NP} \subseteq
\textnormal{DTIME}(2^{O(\log ^2 N)})$, there is no polynomial time algorithm to
distinguish between the following two cases,

\begin{itemize}
\item \textnormal{{\bf YES} Case.} There is an assignment to the variables of $\mathcal{A}$ that
satisfies $1 - c(n) := 1 - 2^{-\Omega(\sqrt{\log n})}$ fraction of the equations
(completeness).

\item \textnormal{{\bf NO} Case.} No assignment to the variables of $\mathcal{A}$ satisfies more than
$1 - s(n) := 1 - \Omega(\log^{-3}n)$ fraction of the equations (soundness).
\end{itemize}
\end{theorem}
The usefulness of the above theorem is due to the fact that the
completeness is very close to $1$, while the soundness is bounded away
from $1$ to allow for $\poly(\log n)$ rounds of parallel repetition.

The rest of this section describes PCP constructions --
required for Theorems
\ref{thm-main2} and \ref{thm-main3} -- which are somewhat more
complicated.

\subsection{Multi-Layered PCP}\label{sec-multi}
The Multi-Layered PCP described here was constructed by Dinur,
Guruswami, Khot and Regev~\cite{DGKR03}  
who also proved its useful properties. An instance $\Phi$ of the
Multi-Layered PCP is parametrized by
integers $L, R > 1$. The PCP consists of $L$ sets of variables $V_1,
\dots, V_L$. The label set (or range) of the variables in the $l^\textrm{th}$
set $V_l$ is a set $R_l$ where $|R_l| = R^{O(L)}$. For any two
integers $1 \leq l < l' \leq L$, the PCP has a set of constraints
$\Phi_{l,l'}$ in which each constraint depends on one variable $v \in
V_l$ and one variable $v' \in V_{l'}$. The constraint (if it exists)
between $v \in V_l$ and $v' \in V_{l'}$ ($l < l'$) is denoted and
characterized by a projection 
$\pi_{v\rightarrow v'} : R_l\to R_{l'}$. A labeling to $v$
and $v'$ satisfies the constraint $\pi_{v\rightarrow v'}$ if the 
projection (via $\pi_{v\rightarrow v'}$) of the label
assigned to $v$ coincides with the label assigned to $v'$. 

The following useful `weak-density' property of the Multi-Layered PCP was
defined in \cite{DGKR03}, which (roughly speaking) states that any
significant subset of variables induces a significant fraction of
the constraints between some pair of layers. 
\begin{definition}\label{def-weakly-dense}
An instance $\Phi$ of the Multi-Layered PCP with $L$
layers is \textnormal{weakly-dense} if for any $\delta > 0$, given $m
\geq \lceil\frac{2}{\delta}\rceil$ layers $l_1 < l_2 < \dots < l_m$
and given any sets $S_i\subseteq V_{l_i}$, for $i\in[m]$ such that
$|S_i|\geq \delta|V_{l_i}|$; there always exist two layers
$l_{i'}$ and $l_{i''}$ such that the constraints between the variables
in the sets $S_{i'}$ and $S_{i''}$ is at least
$\frac{\delta^2}{4}$ fraction of the constraints between the sets 
$V_{l_{i'}}$ and $V_{l_{l''}}$.
\end{definition}

The following inapproximability of the Multi-Layered PCP was proven by
Dinur et al. \cite{DGKR03} based on the PCP Theorem (\cite{AS},
\cite{ALMSS}) and Raz's Parallel Repetition Theorem (\cite{Raz}).
\begin{theorem}\label{thm-multi}
There exists a universal constant $\gamma >0$ such that for any
parameters $L > 1$ and $R$, there is a weakly-dense $L$-layered PCP
$\Phi = \cup \Phi_{l,l'}$ such that it is NP-hard to distinguish
between the following two cases:
\begin{itemize}
\item \textnormal{{\bf YES} Case:} There exists an assignment of
labels to the
variables of $\Phi$ that satisfies all the constraints.
\item \textnormal{{\bf NO} Case:} For every $1\leq l < l' \leq L$, 
not more that
$1/R^\gamma$ fraction of the constraints in $\Phi_{l,l'}$ can be
satisfied by any assignment. 
\end{itemize}
\end{theorem}

\subsection{The $d$-to-$1$ Games Conjecture}\label{sec-dto1}
Before we state the conjecture we need to define a $d$-to-$1$
Game.
\begin{definition} For a positive integer $d$, a 
$d$-to-$1$ Game $\mc{L}$ consists 
two sets of variables $\mc{U}$ and $\mc{V}$, label sets $[k]$ and
$[m]$, and set of constraints $\mc{E}$ where 
each constraint $\pi_{v\rightarrow u} : [m] \rightarrow [k]$ is 
between a variable $v \in \mc{V}$ and $u \in  \mc{U}$, and for any $i
\in [k]$ $\left|\pi_{v \rightarrow u}^{-1}(i)\right| = d$. A
labeling $\sigma$ to the variables in $\mc{U}$ from $[k]$ and $\mc{V}$
from $[m]$ satisfies a constraint  $\pi_{v\rightarrow u}$ iff
$\pi_{v\rightarrow u}(\sigma(v)) = \sigma(u)$. 
\end{definition}
Note that the definition of $d$-to-$1$ Game in \cite{Khot02} had the
condition that $\left|\pi_{v \rightarrow u}^{-1}(i)\right|
\leq d$. All of our proofs go through analogously with this relaxed
condition, but to avoid notational complications we stick to assuming
that the pre-image of every singleton is of size exactly $d$. We now state
the $d$-to-$1$ Games Conjecture.
\begin{conjecture} \label{conj-dto1bireg} \textnormal{($d$-to-$1$ Games
Conjecture~\cite{Khot02})} 
There is a fixed positive integer $d$ such that for
any $\zeta > 0$, there exist integers $k$ and $m$ such that given a
$d$-to-$1$ Game instance $\mc{L}$ with label sets $[k]$ and $[m]$ it
is NP-hard to distinguish between the following two cases:

\begin{itemize}
\item \textnormal{{\bf YES} Case.} There is a labeling to the variables that satisfies all the
constraints.

\item \textnormal{{\bf NO} Case.} Any labeling to the variables satisfies at most $\zeta$
fraction of constraints.
\end{itemize}
In addition we make the assumption\footnote{It is not known whether
this assumption can be made WLOG. However, all known Label Cover
constructions are bi-regular which makes the assumption, in the
authors' opinion, a reasonable one.}  that the instance $\mc{L}$ is
\emph{bi-regular}, i.e.
for any variable $v \in \mc{V}$ 
the number of constraints containing $v$ is the same, and similarly 
for any variable $u \in \mc{U}$ 
the number of constraints containing $u$ is the same. 
\end{conjecture}
Using Conjecture \ref{conj-dto1bireg} we have the following
layered PCP with an additional \emph{smoothness} property.

\subsection{Smooth $d$-to-$1$ Multi-Layered
PCP}\label{sec-dto1multidef}
The following is an analogue of the Multilayered PCP based on the
$d$-to-$1$ conjecture and also incorporating the \emph{smoothness}
property. We shall refer to it as the \emph{Smooth $d$-to-$1$ MLPCP}. 

An instance $\Phi$ of the
Smooth $d$-to-$1$ MLPCP is parametrized by
integers $d, L, R, T > 1$. The PCP consists of $L$ sets of variables $V_1,
\dots, V_L$. The label set (or range) of the variables in the $l^\textrm{th}$
set $V_l$ is a set $R_l$ where $|R_l| = R^{O(TL)}$. For any two
integers $1 \leq l < l' \leq L$, the PCP has a set of constraints
$\Phi_{l,l'}$ in which each constraint depends on one variable $v \in
V_l$ and one variable $v' \in V_{l'}$. The constraint (if it exists)
between $v \in V_l$ and $v' \in V_{l'}$ ($l < l'$) is denoted and
characterized by a projection 
$\pi_{v\rightarrow v'} : R_l\to R_{l'}$. The projection
$\pi_{v\rightarrow v'}$ has the property that for every $j \in
R_{l'}$, $\left|\pi_{v\rightarrow v'}^{-1}(j)\right| = d^{l-l'}$. 
A labeling to $v$
and $v'$ satisfies the constraint $\pi_{v\rightarrow v'}$ if the 
projection (via $\pi_{v\rightarrow v'}$) of the label
assigned to $v$ coincides with the label assigned to $v'$. 

We have a similar weak density property as in the previous section.
\begin{definition}\label{def-weakly-dense-dto1}
An instance $\Phi$ of  Smooth $d$-to-$1$ MLPCP with $L$
layers is \textnormal{weakly-dense} if for any $\delta > 0$, given $m
\geq \lceil\frac{2}{\delta}\rceil$ layers $l_1 < l_2 < \dots < l_m$
and given any sets $S_i\subseteq V_{l_i}$, for $i\in[m]$ such that
$|S_i|\geq \delta|V_{l_i}|$; there always exist two layers
$l_{i'}$ and $l_{i''}$ such that the constraints between the variables
in the sets $S_{i'}$ and $S_{i''}$ is at least
$\frac{\delta^2}{4}$ fraction of the constraints between the sets 
$V_{l_{i'}}$ and $V_{l_{l''}}$.
\end{definition}

We also have the \emph{smoothness} property as defined below.
\begin{definition}\label{def-smoothness}
An instance $\Phi$ of Smooth $d$-to-$1$ MLPCP with $L$ layers and
parameter $T$ has the \emph{smoothness} property if for any two layers
$l < l'$, and variable $v \in V_l$ and two distinct labels $i, j \in
R_l$,
$$\Pr_{v' \in N(v)\cap V_{l'}}\left[\pi_{v\rightarrow v'}(i) =
\pi_{v\rightarrow v'}(j)\right] \leq \frac{1}{T},$$
where the probability is taken over a random variable in $V_{l'}$
which has a constraint with $v$. 
\end{definition}

The following inapproximability of the Smooth $d$-to-$1$ MLPCP
essentially follows
from combining Conjecture \ref{conj-dto1bireg} with the layered 
construction of \cite{Khot-3}. A sketch of the construction is provided in
Section \ref{sec-dto1multi}.
\begin{theorem}\label{thm-dto1multi} Assuming Conjecture
\ref{conj-dto1bireg} the following holds.
There exists a universal constant positive integer $d$ such that for
any arbitrarily small constant $\zeta > 0$, there exists a positive
integer $R$, such that for every $L, T > 1$,
there is a weakly-dense  and smooth $L$-layered PCP with parameters
$d, T, R$,
$\Phi = \cup \Phi_{l,l'}$, such that it is NP-hard to distinguish
between the following two cases:
\begin{itemize}
\item \textnormal{{\bf YES} Case.} There exists an assignment of
labels to the
variables of $\Phi$ that satisfies all the constraints.
\item \textnormal{{\bf NO} Case.} For every $1\leq l < l' \leq L$, 
not more that
$\zeta$ fraction of the constraints in $\Phi_{l,l'}$ can be
satisfied by any assignment. 
\end{itemize}
\end{theorem}

\section{Independent Set in 
Almost $2$-Colorable $4$-Uniform Hypergraphs}\label{sec-main1}

This section presents a hardness reduction from Theorem \ref{thm-KP}
to an instance of {\sc ISAlmostColor}$_\eps(4, 2, Q)$. The reduction
employs an Inner Verifier based on Hadamard Codes. The Hadamard Code
of an element $v \in \Ft^m$ is a $\Ft$-valued code 
indexed by the elements of $\Ft^m$ and
its value at $x \in \Ft^m$ is the dot-product $x\cdot v \in \Ft$. 

\subsection{Hardness Reduction}
Let $\mc{A}$ be the {\sc Max-$3$Lin} instance given by Theorem
\ref{thm-KP}.  
The reduction begins with choosing a positive integer $r$ 
which we shall set later. In the first part of the reduction we shall
construct an Outer Verifier which shall be an $r$-round parallel
repetition of a verifier-prover game obtained from the instance
$\mc{A}$. 

\subsubsection{Outer Verifier}
Let $\Phi_r$ be the collection of all blocks
of $r$ variables each from $\mc{A}$, and $\Psi_r$ be the collection of all
blocks of $r$ equations each. 

Consider the following $2$-prover $1$-round game {\sc $2$P$1$R}$(\mc{A}, r)$:
\begin{enumerate}
\item The Verifier chooses one block $W$ uniformly at random from
$\Psi_r$. From each equation in $W$, the verifier chooses one out of
the three variables at random to construct a block $U$ of $\Phi_r$. 
\item The Verifier sends $U$ to Prover-1 and $W$ to 
Prover-2 and expects from each prover an assignment to all the
variables that it received.
\item The Verifier accepts if the assignment given by Prover-2
satisfies all equations of $W$ and is consistent with the assignment
given to the variables of $U$ by Prover-1.
\end{enumerate}
The Parallel Repetition Theorem of Raz~\cite{Raz} and its
subsequent strengthening by Holenstein~\cite{Hol} and Rao~\cite{Rao}
imply the following.
\begin{theorem}\label{thm-2P1R} The $2$ prover $1$ round game 
{\sc $2$P$1$R}$(\mc{A}, r)$, where $\mc{A}$ is an instance on
$n$ variables
 given by
Theorem \ref{thm-KP}, has the following properties:
\begin{itemize}
\item \textnormal{{\bf YES} Case.} If $\mc{A}$ is a YES instance then the 
Verifier accepts with
probability at least $(1 - c(n))^r$. 

\item \textnormal{{\bf NO} Case.} If $\mc{A}$ is a NO instance then the Verifier accepts with
probability at most $(1 - s(n)^\kappa)^{r/\kappa}$ for some universal
constant $\kappa > 1$.
\end{itemize}
\end{theorem}  
For the rest of the reduction we shall assume that none of the
blocks $W$ or $U$ contain a repeated variable. This omits only a tiny
fraction of blocks which does not change any parameter noticeably. 

\subsubsection{Inner Verifier}
Consider a block $W$ of $r$ equations. It contains $3r$ distinct
variables say $x_1, x_2 \dots, x_{3r-1}, x_{3r}$. 
We may assume without loss of generality 
that the $i$th equation consists of
the variables $x_{3i-2}, x_{3i-1}$, and $x_{3i}$, for $i=1, \dots, r$.
We shall now associate an element of $\Ft^{3r+1}$ with each of the
$r$ equations of $W$. Note that the $(3r+1)$th coordinate is extra and
added to help with ensuring consistency.

Suppose that the $i$th (for some $i \in [r]$) equation is of the form
$x_{3i-2} + x_{3i-1} + x_{3i} = 0$, then let $h_i \in \Ft^{3r+1}$ be
such that  the dot-product 
$h_i\cdot x =  x_{3i-2} + x_{3i-1} + x_{3i}$ for any $x \in
\Ft^{3r + 1}$. Otherwise, if the $i$th equation is of the form 
$x_{3i-2} + x_{3i-1} + x_{3i} = 1$, then let $h_i$ be such that
$h_i\cdot x = x_{3i-2} + x_{3i-1} + x_{3i} + x_{3r + 1}$. Our
assumption that the block $W$ does not contain a repeated variable
implies that the set of elements $\{h_i\}_{i=1}^r$ is linearly
independent. Let $H_W$ be the $r$ dimensional space spanned by 
$\{h_i\}_{i=1}^r$. For completing the reduction we also 
define an element $h_W \in \Ft^{3r+1}$ so that $h_W\cdot x = x_{3r+1}$
for any $x \in \Ft^{3r+1}$.  

Let $\ol{C}[W]$ be a $\{0,1\}$ \emph{code} indexed by the elements of
$\Ft^{3r+1}/H_W$, i.e. the set of cosets of the subspace $H_W$ in the space
$\Ft^{3r+1}$. Since $H_W$ is a $r$ dimensional subspace, the size of
the code $\ol{C}[W]$ is $2^{2r+1}$. We say that $\ol{C}[W]$ is 
\emph{folded} over $H_W$. 
It is easy to see that any $\ol{C}[W] :  \Ft^{3r+1}/H_W \mapsto \{0,1\}$ can
be \emph{unfolded} into $C[W] : \Ft^{3r+1} \mapsto \{0,1\}$ such that,
$C[W](x + y) = \ol{C}[W](x + H_W)$ for any $x \in \Ft^{3r+1}$ and $y
\in H_W$.
For notational convenience we shall represent the coset $x + H_W$ simply
by $x$, and this shall be clear from the context. 

Ideally, $C[W]$ is supposed to be the Hadamard Code of a satisfying 
assignment to the variables in $W$ with the $(3r+1)$th coordinate set to $1
\in \Ft$, so that the code $\ol{C}[W]$ is well defined and 
folded over the subspace $H_W$.

We are now ready to define the vertices and hyperedges of the instance
$G(V, E)$ of \IS$(2,4)$. 

\medskip
\noindent
{\bf Vertices.} The vertex set $V$ consists of all the locations of
$\ol{C}[W]$ for each $W \in \Psi_r$, i.e. each block $W$ of $r$ equations.

\medskip
\noindent
{\bf Hyperedges.} Consider any choice of $U \in \Phi_r$ and $W \in
\Psi_r$ by the verifier in the game {\sc $2$P$1$R}$(\mc{A}, r)$ in
Step 1. Let $U$ and $W' \in \Psi_r$ be another choice with the same
block of $r$ variables $U$. Let $\pi_W : \Ft^{3r+1} \mapsto \Ft^r$ be
a projection onto the coordinates of the $r$ variables of $U$ from
the block of $(3r+1)$ coordinates corresponding to the 
$3r$ variables of $W$ and the extra coordinate as defined above.
The extra coordinate plays no part in this projection. We
shall also use the notation $\pi^{-1} : \Ft^r \mapsto \Ft^{3r + 1}$,
which extends a vector by filling in zeros in the rest of the
coordinates.
Similarly, $\pi_{W'}$ be the projection for $W'$. Let $\ol{C}[W]$ and
$\ol{C}[W']$ be the codes of $W$ and $W'$.
For all such choices of $U$, $W$ and $W'$ do the following.
\begin{enumerate}
\item For all choices of elements $x, y \in \Ft^{3r + 1}$ and
$z \in \Ft^r$ such that $z \neq 0$, do step 2. 
\item Add a hyperedge between the vertices (or locations of the
codes): $\ol{C}[W](x), \ol{C}[W](x + \pi_W^{-1}(z) + h_W), 
\ol{C}[W'](y)$ and $\ol{C}[W'](y + \pi_{W'}^{-1}(z))$.
It is easy to see that since $z \neq 0$ the four vertices chosen above are
distinct.
\end{enumerate}
This completes the hardness reduction and we move to its analysis.

\subsection{YES Case}
In the YES Case the instance $\mc{A}$ has an assignment $\sigma^*$ to its
variables that satisfies $(1 - c(n))$ fraction of its equations. Call
the equations satisfied by $\sigma^*$ as \emph{good}. Similarly, call
a block of $r$ equations as \emph{good} if all of its equations are
good. Clearly, at least $(1 - c(n))^r$ fraction of the blocks are
good. 

For any good block $W$ let $C[W] : \Ft^{3r+1} \mapsto \Ft$ 
be the Hadamard Code of the
assignment $\sigma^*(W)\in \Ft^{3r}$ to the variables in $W$,
concatenated with a $1$ in the $(3r+1)$th coordinate. Let us denote
this concatenated vector as $(\sigma^*(W), 1)$.
In other words, $C[W](x) =
(\sigma^*(W), 1)\cdot x \in \Ft$, for $x \in \Ft^{3r+1}$. 
Since $\sigma^*$ satisfies all equations in $W$, it is easy to see that
it is invariant over the cosets of $H_W$, i.e. $C[W](x + y) = C[W](x)$
for $x \in \Ft^{3r+1}$ and $y \in H_W$. Thus this can be folded into
the code $\ol{C}[W]$ by defining $\ol{C}[W](x+H_W) = C[W](x)$. As
before, we shall use $\ol{C}[W](x)$ to represent the value over the
coset $x + H_W$.

The above defines a $2$-coloring of the locations of the codes of all
good blocks depending on its value in $\Ft$. We shall show that any 
hyperedge completely induced by these locations is non-monochromatic.  

Consider a choice of $U, W$ and $W'$ in the construction of the
hyperedges where $W$ and $W'$ are good blocks. Let $x, y$ and $z$ be
chosen as in Step 1. We shall show that,
\begin{equation}
\ol{C}[W](x)+ \ol{C}[W](x + \pi_W^{-1}(z) + h_W)+
\ol{C}[W'](y) + \ol{C}[W'](y + \pi_{W'}^{-1}(z)) = 1,
\end{equation}
which implies that the corresponding hyperedge is non-monochromatic.
To see this, observe that the LHS of the above equation is,
\begin{eqnarray}
& & (\sigma^*(W),1)\cdot x + (\sigma^*(W),1)\cdot(x + \pi_W^{-1}(z) + h_W) +
(\sigma^*(W'),1)\cdot y + (\sigma^*(W'),1)\cdot(y + \pi_{W'}^{-1}(z))
\nonumber \\
& = & (\sigma^*(W),1)\cdot (\pi_W^{-1}(z)) + (\sigma^*(W'),1)\cdot
(\pi_{W'}^{-1}(z)) + (\sigma^*(W),1)\cdot h_W \nonumber \\
& = & (\pi_W(\sigma^*(W)) + \pi_{W'}(\sigma^*(W'))\cdot z + 1
\nonumber \\
& = & 1, 
\end{eqnarray}
where the second last equation follows from the definition of
$h_W$ and
last equation follows from the fact that $\sigma^*$
is a global assignment so its projection onto $U$ from $W$ or
$W'$ is the same.  

Thus, after removing a $1 - (1-c(n))^r$ fraction of vertices
corresponding to the blocks which are not good and all
hyperedges incident on them, the rest of the hypergraph is
$2$-colorable.


\subsection{NO Case}\label{sec-multi-No}
Let $\mc{I}$ be an independent set in $G$. For every block $W$ of $r$
equations, let $\ol{C}[W]$ be the indicator of $\mc{I}$ restricted to
the locations of the code $\ol{C}[W]$. Here, $\ol{C}[W]$ is a thought
of as a $\{0,1\}$ real valued code.

Let $U, W$ and $W'$ be the choices in the construction of the
hyperedges. For all choices of $x, y$ and $z$ in Step 1 of the
construction, we have.
\begin{equation}
\ol{C}[W](x)\cdot \ol{C}[W](x+ \pi_W^{-1}(z) + h_W)\cdot
\ol{C}[W'](y) \cdot \ol{C}[W'](y + \pi_{W'}^{-1}(z)) = 0. 
\end{equation}
As mentioned earlier, we can unfold the codes into $C[W]$ and $C[W']$
to rewrite the above as,
\begin{equation}
C[W](x)\cdot C[W](x+ \pi_W^{-1}(z) + h_W)\cdot
C[W'](y)\cdot C[W'](y + \pi_{W'}^{-1}(z)) = 0. 
\end{equation}
For convenience of notation, we shall refer to $C[W]$ as $A$ and
$C[W']$ as $B$.
Doing the usual Fourier expansion and using standard tools from
Fourier Analysis over folded codes (refer to Section \ref{sec-Fourier}
for an overview) we get the following.
\begin{eqnarray}
& & \displaystyle\sum_{\substack{\alpha, \alpha', \beta, \beta' \in
\Ft^{3r+1}\\ \alpha,
\alpha'\perp H_W \\
\beta, \beta'\perp H_{W'}}}\wh{A}_\alpha\chi_\alpha(x)
\wh{A}_{\alpha'}\chi_{\alpha'}(x+ \pi_W^{-1}(z) + h_W)
\wh{B}_{\beta}\chi_\beta(y)
\wh{B}_{\beta'}\chi_{\beta'}(y + \pi_{W'}^{-1}(z)) = 0 \nonumber \\
&\Rightarrow & 
\displaystyle\sum_{\substack{\alpha, \alpha', \beta, \beta' \\ \alpha,
\alpha'\perp H_W \\
\beta, \beta'\perp H_{W'}}}\wh{A}_\alpha \wh{A}_{\alpha'}\wh{B}_{\beta}
\wh{B}_{\beta'}
\chi_{(\alpha+\alpha')}(x)\chi_{\alpha'}(h_W)\chi_{\pi_W(\alpha')}(z)
\chi_{(\beta+\beta')}(y)\chi_{\pi_{W'}(\beta')}(z) = 0 \label{eqn-zneq0} 
\end{eqnarray}
The above is true for all $x, y$ and $z$ such that $z \neq 0$ 
which are independent
of each other. Thus, for a fixed value of
$x$ and $y$, the expectation of the LHS of Equation
\eqref{eqn-zneq0} over all $z \in \Ft^r$ is equal to $2^{-r}$ times
its value at $z = 0$. Observing that in the expectation over 
all $z \in \Ft^r$ only
terms satisfying $\pi_W(\alpha') = \pi_{W'}(\beta')$ survive, we obtain,
\begin{eqnarray}
& & \displaystyle\sum_{\substack{ \alpha,
\alpha'\perp H_W \\
\beta, \beta'\perp H_{W'} \\ 
\pi_W(\alpha') = \pi_{W'}(\beta')}}\wh{A}_\alpha \wh{A}_{\alpha'}\wh{B}_{\beta}
\wh{B}_{\beta'}
\chi_{(\alpha+\alpha')}(x)\chi_{\alpha'}(h_W)
\chi_{(\beta+\beta')}(y) \nonumber \\
& = & 2^{-r}
\sum_{\substack{ \alpha,
\alpha'\perp H_W \\
\beta, \beta'\perp H_{W'}}}\wh{A}_\alpha \wh{A}_{\alpha'}\wh{B}_{\beta}
\wh{B}_{\beta'}
\chi_{(\alpha+\alpha')}(x)\chi_{\alpha'}(h_W)
\chi_{(\beta+\beta')}(y). 
\end{eqnarray}
Taking a further expectation over $x$ and $y$, we observe
that the only terms that survive on the LHS are those in which $\alpha
= \alpha'$, $\beta = \beta'$ and $\pi_W(\alpha) = \pi_{W'}(\beta)$,
while the terms that survive on the RHS have
$\alpha
= \alpha'$ and $\beta = \beta'$. 
Thus we obtain,
\begin{equation}
\displaystyle\sum_{\substack{\alpha\perp H_W,
\beta \perp H_{W'} \\ \pi_W(\alpha) = \pi_{W'}(\beta)}}
\wh{A}_\alpha^2\wh{B}_\beta^2\chi_\alpha(h_W) = 
2^{-r}\displaystyle\sum_{\substack{\alpha\perp H_W,
\beta \perp H_{W'}}}
\wh{A}_\alpha^2\wh{B}_\beta^2\chi_\alpha(h_W) \leq 2^{-r},
\end{equation}
where the last inequality is because the sum of squares of the Fourier
coefficients is at most $1$.
Now, $\chi_\alpha(h_W) = -1$ if $\alpha\cdot h_W = 1$ and $1$
otherwise. Thus,
\begin{eqnarray}
\displaystyle\sum_{\substack{\alpha\perp H_W,
\beta \perp H_{W'} \\ \pi_W(\alpha) = \pi_{W'}(\beta) \\ \alpha\cdot
h_W = 1}}
\wh{A}_\alpha^2\wh{B}_\beta^2 & \geq & 
\displaystyle\sum_{\substack{\alpha\perp H_W,
\beta \perp H_{W'} \\ \pi_W(\alpha) = \pi_{W'}(\beta) \\ \alpha\cdot
h_W = 0}}
\wh{A}_\alpha^2\wh{B}_\beta^2 - 2^{-r} \nonumber \\
& \geq & \wh{A}_\emptyset^2\wh{B}_\emptyset^2 - 2^{-r}.
\end{eqnarray}
The above gives a strategy for the provers of {\sc $2$P$1$R}$(\mc{A},
r)$. Suppose Prover-1 receives a block of variables $U$ and
Prover-2 receives a block of equations $W$. 

\noindent
Strategy of
Prover-2: It chooses a vector $\alpha \in \Ft^{3r+1}$ 
satisfying: (i) $\alpha \perp H_W$, and (ii) 
$\alpha\cdot h_W = 1$ with probability $\wh{A}_\alpha^2$, where $A =
C[W]$. Since
$\alpha$ satisfies (i) and (ii), the first $3r$ coordinates give a
satisfying assignment to the variables in $W$.
This assignment is returned to the Verifier by Prover-2.

\noindent
Strategy of
Prover-1: It chooses a block $W'$ from the choice of the
verifier of  
{\sc $2$P$1$R}$(\mc{A}, r)$ conditioned on the block of variables
picked being $U$. It then chooses
 $\beta \in \Ft^{3r+1}$ 
satisfying: $\beta \perp H_{W'}$ with probability $\wh{B}_\beta^2$
where $B = C[W']$. The assignment to the variables in $U$ contained in the
first $3r$ coordinates of $\beta$ is returned to the Verifier. 

\medskip
Suppose that the independent set $\mc{I}$ contains $\delta$ fraction
of the vertices of $G$, i.e. locations of the codes. 
Recall that we set the value of the code $\ol{C}[W]$ to be the
indicator of $\mc{I}$ restricted to its locations. Thus, for at least
$\delta/2$ fraction of the blocks $W$, $\E_x[\ol{C}[W](x)] \geq
\delta/2$. Call such blocks as \emph{heavy}.

Conditioned on the block of variables $U$, let $p_U$ be the fraction of choices
of block of equations $W$ by the verifier {\sc $2$P$1$R}$(\mc{A}, r)$ 
such that $W$ is heavy. From the above $\E_U[p_U] \geq \delta/2$, by
the regularity of $\mc{A}$. Thus, the probability that both $W$ and $W'$
are heavy -- where $W'$ is obtained from the strategy of
Prover-1 -- is $E_U[p_U^2] \geq \E_U[p_U]^2 \geq \delta^2/4$.
Noting that the weight of $\ol{C}[W]$ is same as that of $C[W]$ which
is given by the empty coefficient of the Fourier expansion, we obtain
that the verifier accepts with probability at least,
$$\frac{\delta^2}{4}\left(\frac{\delta^2}{4} - 2^{-r}\right)
\geq \left(\frac{\delta^2}{4} - 2^{-r}\right)^2.$$
From Theorem \ref{thm-2P1R} this implies that $\delta^2/4 \leq
(1-s(n)^\kappa)^{r/2\kappa} + 2^{-r}$. 

\medskip
{\bf Setting the Parameters.} We set $r = \log^\ell n$ for a large
enough constant $\ell$. The size of the hypergraph is 
$N = 2^{\poly(\log n)}$. In the YES case, the number of vertices to be
removed is at most $c(n)r \leq 2^{-(\log N)^\xi}$ for some positive
constant $\xi$ (depending on $\ell$). Further, we obtain that 
$(1-s(n)^\kappa)^{r/2\kappa} + 2^{-r} \leq 2^{-(\log N)^{1-\gamma}}$ for
an arbitrarily small $\gamma$ by an appropriately large choice of the
constant $\ell$. The above analysis yields a
bound of $2^{-(\log N)^{1-\gamma}}$ on the relative size of the
largest independent set in the NO case, for arbitrarily small $\gamma
> 0$. 

\section{Independent Set in Almost $2$-colorable
$3$-uniform hypergraphs}\label{sec-main2}

We first need a few useful definitions and results for our analysis
which follows a pattern similar to previous works~\cite{DKPS, KS12,
SS13} and we shall use their notation.

\subsection{Preliminaries}

A family $\mc{F} \subseteq \{*, 1, 2\}^m$ is called \emph{monotone} if
for any $F \in \mc{F}$ and $F'$ obtained by changing a $*$ to either
$1$ or $2$ in any coordinate, $F'\in \mc{F}$. 
For a parameter $p \in [0,1]$, define the measure $\mu_p$ on $\{1, 2,
*\}^m$ by $\mu_{p}(F) = p^{m-m'}(1-p)^{m'}$, where $m'$ is
the number of coordinates of $F$ with $*$ in them, for any $F \in \{1, 2,
*\}^m$. In other words $\mu_p$ is the product measure assigning in
each coordinate a measure $1-p$ to $*$ and $\frac{p}{2}$ to each
of $1$ and $2$. The measure of a family $\mc{F} \subseteq \{1,2,*\}^m$ is
$\mu_p(\mc{F}) = \sum_{F \in \mc{F}}\mu_p(F)$.

A set $C \subseteq [m]$ is a $(\delta, p)$-{\em core} for a family
$\mc{F}$, if there exists a family $\mc{F}'$ such that
$\mu_p(\mc{F}\triangle\mc{F}' )\leq \delta$ and ${\cal F}'$
depends only on the coordinates in $C$.
Let $t \in (0,1)$ be a given parameter 
and $C \subseteq [m]$. A \emph{core-family} $[\mc{F}]^t_C$ is a family
on the set of coordinates $C$ which resembles $\mc{F}$ restricted to $C$.
Formally,
$$[\mc{F}]^t_C \defeq \left\{ F\in \{*, 1, 2\}^C\ \middle|\ \Pr_{F' \in
    \mu_p^{[m]\setminus C}}\left[(F, F')\in {\cal F}\right] >
  t\right\},$$ where $(F, F')$ is an element in $\{*,1,2\}^m$ by
combining $F$ on coordinates in $C$ and
$F'$ on $[m]\setminus C$. The \emph{influence} of a coordinate $i \in
[m]$ for a family ${\cal F}$ under the measure $\mu_p$ 
is defined as follows:
$$\tn{Inf}^p_i({\cal F}) := \mu_p\left(\left\{F : 
F\mid_{i=*}\not\in {\cal
      F}\textnormal{ and }F\mid_{i=j}\in {\cal F}\textnormal{ for some
    } j\in \{1,2\}\right \}\right),$$ where $F\mid_{i=*}$ is an element
identical to $F$ except on the $i^\textrm{th}$ coordinate where it is
$*$, and $F\mid_{i=r},$ for $r \in \{1,2\}$ is similarly defined.  The 
  \emph{average sensitivity} of ${\cal F}$ at $p$ is the sum of
influence of all coordinates:
$\tn{as}_{p}({\cal F}) := \sum_{i=1}^m\tn{Inf}^p_i({\cal F})$.

Let $D^p$ be a distribution on $\{*, 1, 2\}^2$ defined by
first sampling $(1, 2)$ and $(2,1)$ uniformly with probability
$\slfrac{1}{2}$ each and then changing each
coordinate to $*$ independently with probability $1 - p$. It is easy to
see that both the marginals of $D^p$ are identical to
$\mu_p$. 

\subsubsection{Useful Results}
The following variant of Russo's Lemma was proved in
\cite{DKPS} (as Lemma 1).
\begin{lemma}
[Russo's Lemma \cite{Russo}]
\label{lem-russo}
Let ${\cal F}\subseteq \{*, 1, 2\}^m$ be monotone, then
$\mu_p({\cal F})$ is increasing with $p$. In fact,
$$\frac{1}{2}\cdot \tn{as}_p({\cal F}) \leq \frac{d\mu_p({\cal
F})}{dp} \leq \tn{as}_p({\cal F}).$$
\end{lemma}
The following corollary follows from the above and is proved in
\cite{SS13}.
\begin{corollary}\label{cor-follow}
For a monotone family $\mc{F}\subseteq \{*, 1,2\}^m$,
\begin{enumerate}
\item For any $p' \geq p$, $\mu_{p'}({\cal F}) \geq \mu_{p}({\cal
    F})$.
\item For any $\eps  > 0$, there is a $p' \in [1 - \eps, 1 - \eps/2]$
such that $\tn{as}_{p'}({\cal F}) \leq \frac{4}{\eps}$.  
\end{enumerate}
\end{corollary}
The following is a generalization of Friedgut's Junta Theorem which is
proved in \cite{ST11}.
\begin{theorem}[Friedgut's Theorem \cite{Friedgut, ST11}]
\label{thm-Friedgut}
Fix $\delta > 0$. Let $\mc{F} \subseteq \{*, 1, 2\}^m$ be
monotone with $a = \tn{as}_p({\cal F})$, for $p \in [0,1]$. 
There exists a function
$C_{Friedgut}(p, \delta, a) \leq c_p^{a/\delta}$, for a constant
$c_p$
depending only on $p$, so that $\mc{F}$ has a $(\delta, p)$-core $C$
of size $|C| \leq C_{Friedgut}(p, \delta, a)$.
\end{theorem}
The above theorem shall be used along with the
following generalization of Lemma
3.1 in \cite{DS} proved in \cite{SS13}.
\begin{proposition}\label{prop-coremass} 
If $C$ is a $(\delta, p)$-core of $\mc{F}$, then
$\mu_p^C\left([\mc{F}]_C^{\slfrac{3}{4}}\right) \geq \mu_p({\cal F}) - 
3\delta.$
\end{proposition}

Using the above one can prove the following lemma.
\begin{lemma} \label{lem-2element}
For a fixed parameter $p \in (0,1)$ and a positive constant $\delta$, 
given a monotone family $\mc{F}\subseteq
\{*,1,2\}^m$ such that
$\mu_p(\mc{F}) \geq \delta$, there exists a subset $S \subseteq [m]$
such that $|S|\leq \bar{c}_p^{\slfrac{16}{(1-p)\delta}}$, for some
constant $\ol{c}_p$ depending only on $p$, and two elements $F, F' \in
\mc{F}$ such that for all $j \not\in
S$, $(F(j), F'(j))$ is not $(1,1)$ or $(2,2)$.
\end{lemma}
\begin{proof}
We first choose $\bar{c}_p = \max\{c_{p'} \mid p' \in [1-\eps,
1-\slfrac{\eps}{2}]\}$ where $p := 1 - \eps$. By Corollary
\ref{cor-follow} there is a $p' \in [1-\eps, 1-\slfrac{\eps}{2}]$ such
that $a := \tn{as}_{p'}(\mc{F}) \leq \slfrac{4}{\eps} =
\slfrac{4}{(1-p)}$. Using Theorem
\ref{thm-Friedgut} one can obtain a $(\slfrac{\delta}{4}, p')$-core
$S$ of $\mc{F}$ of size $|S|\leq \bar{c}_p^{\slfrac{16}{\delta(1-p)}}$. By
Proposition \ref{prop-coremass} and using the fact that
$\mu_{p'}(\mc{F}) \geq \mu_p(\mc{F}) \geq \delta$, we get that,
$$\mu_{p'}^S\left([\mc{F}]_S^{\slfrac{3}{4}}\right)\geq
\delta - \slfrac{3\delta}{4} = \slfrac{\delta}{4} > 0,$$ where 
$[\mc{F}]_S^{\slfrac{3}{4}}$ is the
core-family with respect to the measure $\mu_{p'}$.
Choose an element $\ol{F} \in [\mc{F}]_S^{\slfrac{3}{4}}$.
Probabilistically construct
$F, F' \in \mc{F}$ as follows. For $j \in S$, set $F(j)$ and $F'(j)$
to the corresponding value $\ol{F}(j)$. For $j \not\in S$
independently sample
$(F(j), F'(j))$ from $D^{p'}$. Since the marginals of $D^{p'}$ are
distributed as $\mu_{p'}$, by the definition of a core-family, we
have,
$$\Pr[F \in \mc{F}\textnormal{ and } F' \in\mc{F}] \geq 1 -
2\left(\frac{1}{4}\right) \geq \frac{1}{2}.$$
Moreover, since $(1,1)$ and $(2,2)$ do not lie in the support of
$D^{p'}$, the elements $F, F'$ satisfy the condition of the lemma.
\end{proof}

\subsection{Hardness Reduction}
Let $\delta, \eps > 0$ be parameters that we shall set later. We begin
with an instance $\Phi$ of the Multi-Layered PCP from Theorem
\ref{thm-multi}.
The number of layers $L$ of $\Phi$ is chosen to be $\lceil
32\delta^{-2}\rceil$. The parameter $R$ shall be set later to be large
enough. In the following paragraphs we describe the construction of
a weighted $3$-uniform hypergraph $G$ with vertex
set $\mc{H}$ a hyperedge set $\mc{E}$ and a weight function $\wt$ 
on the vertices,  as an instance of {\sc ISAlmostColor}$_\eps(3, 2,
\slfrac{1}{\delta})$. 

\medskip
\noindent
{\bf Vertices.} Consider a variable $v \in V_l$, i.e. in the 
$l$th layer of $\Phi$. Let a \emph{Long Code} 
$\mc{H}^v$ be a copy of the set
$\{1,2,*\}^{R_l}$ equipped with the measure $\mu_p$ where $p:=1-\eps$. 
The set of
vertices $\mc{H} := \cup_{1\leq l\leq L}\cup_{v\in V_l}\mc{H}^v$. 
The weight of any $x \in \mc{H}^v$ is,
$$\wt(x) = \frac{\mu_p(x)}{L|V_l|}.$$
Thus, the total weight of the vertices corresponding to any layer of
the PCP is $1/L$, which is equally distributed over the Long Codes of
all the variables in that layer.

\medskip
\noindent
{\bf Hyperedges.} For all variables $v \in V_l$ and $u \in V_{l'}$  ($l
< l'$) such that there is a constraint
$\pi_{v\rightarrow u}$ between them, 
add a hyperedge between all $x \in
\mc{H}^{u}$ and $y, z \in \mc{H}^v$ which satisfy the following
condition: 
For any $j \in R_l$ and $i = \pi_{v\rightarrow u}(j) \in R_{l'}$,
the tuple $(x_i, y_j, z_j)$ is not $(1,1,1)$ or $(2,2,2)$.

\subsection{YES Case}
In the YES case, there is an assignment $\sigma$ 
of labels to the variables of
$\Phi$ that satisfies all the constraints. Construct a partition of
$\mc{H}$ into disjoint subsets $\mc{H}_1, \mc{H}_2$ and $\mc{H}_*$ as
follows. For any variable $v$ of $\Phi$, add $x \in \mc{H}^v$ to
$\mc{H}_{x_{\sigma(v)}}$. It is easy to see that $\wt(\mc{H}_*) =
\eps$ and $\wt(\mc{H}_1) = \wt(\mc{H}_2) = \left(\frac{1-\eps}{2}\right)$. 

Furthermore, Let  $v, u$ be variables such that there is a constraint
$\pi_{v\rightarrow u}$ between them. Suppose there is a hyperedge
between $x \in \mc{H}^u$ and $y, z \in \mc{H}^v$. Since $\sigma$ is a
satisfying assignment, $\pi_{v\rightarrow u}(\sigma(v)) = \sigma(u)$.
By the construction of the hyperedges, this implies that the tuple
$(x_{\sigma(u)}, y_{\sigma(v)}, z_{\sigma(v)})$ is not $(1,1,1)$ or
$(2,2,2)$, and thus the hyperedge $(x,y,z)$ is not contained in
$\mc{H}_1$ or in $\mc{H}_2$. Therefore, removing the set of vertices
$\mc{H}_*$ of weight $\eps$ and the hyperedges incident on it makes
the hypergraph $2$-colorable. 

\subsection{NO Case}
In the NO Case assume that there is a maximal independent set
$\mc{I}\subseteq \mc{H}$ of weight $\wt(\mc{I}) \geq \delta$. From the
construction of the hyperedges, it is easy to see that any maximal
independent set is monotone. 
Let $\mc{I}^v := \mc{I}\cap \mc{H}^v$ for any variable
$v$ of $\Phi$. Thus, each $\mc{I}^v$ is a monotone family.

Consider the set of variables 
$$U := \left\{ u \in V \mid \mu_p(\mc{I}^u) = 
\frac{\wt(\mc{I}^u)}{\wt(\mc{H}^u)} \geq
\frac{\delta}{2}\right\}.$$ By averaging, it is easy to see that,
$$\sum_{u\in U}\wt(\mc{H}^u)\geq \frac{\delta}{2}.$$
Another averaging shows that for at least $\frac{\delta}{4}L \geq
\frac{8}{\delta}$ layers $l$, at least $\frac{\delta}{4}$ fraction of
variables in layer $l$ belong to $U$. Applying the weak density
property we obtain two layers $l < l'$ such that at least
$\frac{\delta^2}{64}$ fraction of the constraints between $V_l$ and
$V_{l'}$ are between the variables in $U_l := U\cap V_l$ and $U_{l'} :=
U\cap V_{l'}$. The following key lemma follows from Lemma
\ref{lem-2element}. 
\begin{lemma}\label{lem-key} 
For any variable $v \in U_l$ there is a subset 
$S^v \subseteq R_l$ of size $|S^v| \leq 
t(\eps,\delta) := c_\eps^{\slfrac{1}{\delta}} 
$ for some 
constant $c_\eps
> 0$ depending on $\eps$, 
and elements $y^v, z^v \in \mc{I}^v$ such that for all $j \in
R_l\setminus S^v$, the tuple $(y^v_j, z^v_j)$ is not $(1,1)$ or
$(2,2)$. 
\end{lemma}
Note that in the above, if $S^v$ is empty then $y^v$ and $z^v$ will
trivially ensure a hyperedge in $\mc{I}$, so we may assume it is
non-empty.

Using the above lemma we can now define the labeling for each of the
variables in $U_{l}$ and $U_{l'}$.

\medskip
\noindent 
{\bf Labeling for $v \in U_{l}$:} Choose a label $\rho(v) \in R_l$
uniformly at random from $S^v$. 

\medskip
\noindent
{\bf Labeling for $u \in U_{l'}$:} This choice is made depending on
the labeling of variables in $U_l$. Let $N(u) \subseteq V_l$ be all the
variables in $V_l$ which have a constraint with $u$. Choose a label
$\lambda(u)$ defined below,
$$\lambda(u) := \textnormal{argmax}_{a \in R_{l'}} \left|\{v 
\in N(u)\cap U_l \mid
\pi_{u\rightarrow v}(\rho(v)) = a\}\right|.$$
In other words, $\lambda(u)$ is the label in $R_{l'}$ 
which is the projection of
the maximum number of labels of the neighbors of $u$ in $U_l$. 

\medskip

For the rest of the analysis we shall focus on one variable $u \in
U_{l'}$ and its neighborhood in $U_l$, $N(u)\cap U_l$.
To complete the analysis we need the following lemma proved in
\cite{DGKR03}. 
\begin{lemma}\label{lem-inter} Let $A_1, A_2, \dots, A_N$ be a
collection of $N$ sets, each of size at most $T\geq 1$. If there are
not more than $D$ pairwise disjoint subsets in the collection then
there must exist an element which belongs to at least $\frac{N}{TD}$
sets.
\end{lemma}
Consider the collection $\{\pi_{v\rightarrow u}(S^v) 
\mid v \in N(u)\cap U_l\}$. Each
subset in this collection is of size at most $t(\eps, \delta)$.
Each such subset $\pi_{v \rightarrow u}(S^v)$ 
rules out $\mc{I}^u$ containing any element
$x^u$
with $*$ in all coordinates corresponding to $\pi_{v\rightarrow
u}(S^v)$. Otherwise, $(x^u, y^v, z^v)$ would be a hyperedge in
$\mc{I}$. 
Suppose there are $r$ pairwise disjoint subsets in this collection.
This would reduce the measure $\mu_p(\mc{I}^u)$ by a factor of
$\left(1 - \eps^{t(\eps, \delta)}\right)^r$. However,
$\mu_p(\mc{I}^u) \geq \frac{\delta}{2}$. Thus, $r$ is at most $\log
\left(\frac{\delta}{2}\right)/\log \left(1 - \eps^{t(\eps,
\delta)}\right)$. Applying Lemma \ref{lem-inter} there is an element
$a$ contained in at least 
$$\frac{\log \left(1 - \eps^{t(\eps,
\delta)}\right)}{\left(t(\eps,
\delta)\log\left(\frac{\delta}{2}\right)\right)},$$
fraction of the subsets in the collection $\{\pi_{v\rightarrow
u}(S^v) \mid v \in N(u)\cap
U_l\}$. This implies that in expectation, over the choice of
$\{\rho(v) \mid v \in N(u)\cap U_l\}$, $\pi_{v\rightarrow
u}(\rho(v)) = a$ for at least,
$$\xi(\eps, \delta): = \frac{\log \left(1 - \eps^{t(\eps,
\delta)}\right)}{\left(t(\eps,
\delta)^2\log\left(\frac{\delta}{2}\right)\right)},$$
fraction of $N(u)\cap U_l$. Thus, in expectation the labelings $\rho$
and $\lambda$ satisfy $\xi(\eps,
\delta)\left(\frac{\delta^2}{64}\right)$ fraction of the constraints
between the layers $l$ and $l'$. Choosing the parameter $R$ of $\Phi$
to be small enough gives a contradiction.

\section{Independent Set in $2$-Colorable
$3$-Uniform Hypergraphs}\label{sec-main3}

We begin with a few useful definitions and results, which can also
be found in greater detail in \cite{Mossel}. The correlation between
two correlated probability spaces is defined as follows.
\begin{definition}\label{def-corr}
Suppose $(\Omega^{(1)} \times \Omega^{(2)}, \mu)$ is a finite
correlated probability space with the marginal probability spaces
$(\Omega^{(1)}, \mu)$ and  $(\Omega^{(2)}, \mu)$. The
\emph{correlation} between these spaces is,
$$\rho(\Omega^{(1)}, \Omega^{(2)}; \mu) = \tn{sup}
\left\{\left|\E_{\mu}[fg]\right|
\mid f\in L^2(\Omega^{(1)}, \mu), g\in L^2(\Omega^{(2)}, \mu), 
\E[f]= \E[g] = 0; \E[f^2], \E[g^2] \leq 1\right\}.$$
Let $(\Omega^{(1)}_i\times \Omega^{(2)}_i, \mu_i)_{i=1}^n$ be a sequence
of correlated spaces. Then,
$$\rho(\prod_{i=1}^n\Omega^{(1)}_i, \prod_{i=1}^n\Omega^{(2)}_i; 
\prod_{i=1}^n\mu_i) \leq \max_{i} \rho(\Omega^{(1)}_i, \Omega^{(2)}_i;
\mu_i).$$
Further, the correlation of $k$ correlated spaces
$(\prod_{j=1}^k\Omega^{(j)}, \mu)$ is defined as
follows: 
$$\rho(\Omega^{(1)},\Omega^{(2)}, \dots, \Omega^{(k)}; \mu) := 
\max_{1\leq i\leq k} \rho\left(\prod_{j=1}^{i-1}\Omega^{(j)}\times
\prod_{j=i+1}^k\Omega^{(j)}, \Omega^{(i)}; \mu\right).$$ 
\end{definition}
\begin{lemma}\label{lem-corbd} Let $(\Omega^{(1)}\times\Omega^{(2)},
\mu)$ be two correlated spaces such that the probability of the
smallest atom in $(\Omega^{(1)}\times\Omega^{(2)},
\mu)$ is at least $\alpha \in (0,\slfrac{1}{2}]$. Define a bipartite
graph between $\Omega^{(1)}$ and $\Omega^{(2)}$ with an edge between
$(a,b) \in \Omega^{(1)}\times\Omega^{(2)}$ if $\mu(a,b) > 0$. If this
graph is connected then,
$$\rho(\Omega^{(1)}, \Omega^{(2)};
\mu) \leq 1 - \slfrac{\alpha^2}{2}.$$
\end{lemma}
\noindent 
We shall also refer to the following Gaussian stability measures in our
analysis.
\begin{definition}
Let $\Phi : \R\mapsto [0,1]$ be the cumulative distribution function
of the standard Gaussian. For a parameter $\rho$, define,
$$\underline{\Gamma}_\rho(\mu, \nu) = \Pr[X \leq \Phi^{-1}(\mu), Y \geq
\Phi^{-1}(1- \nu)],$$ 
$$\ol{\Gamma}_\rho(\mu, \nu) = \Pr[X \leq \Phi^{-1}(\mu), Y \leq
\Phi^{-1}(\nu)],$$ 
where $X$ and $Y$ are two standard Gaussian variables with covariance
$\rho$.
\end{definition}
The Bonami-Beckner
operator is defined as follows.
\begin{definition} Given a probability space $(\Omega, \mu)$ and 
$\rho \geq 0$, consider the space $(\Omega\times \Omega, \mu')$ where
$\mu'(x,y) = (1-\rho)\mu(x)\mu(y) + \rho\mathbf{1}_{\{x=y\}}\mu(x)$,
where $\mathbf{1}_{\{x=y\}} = 1$ if $x=y$ and $0$ otherwise. The
Bonami-Beckner operator $T_\rho$ is defined by,
$$ (T_\rho f)(x) = \E_{(X,Y)\leftarrow \mu'}\left[f(Y) \mid X =
x\right].$$
For product spaces $(\prod_{i=1}^n \Omega_i, \prod_{i=1}^n \mu_i)$, the
Bonami-Beckner operator $T_\rho = \otimes_{i=1}^n T^i_\rho$, where
$T^i_\rho$ is the operator for the $i$th space $(\Omega_i, \mu_i)$. 
\end{definition}
\noindent
By Proposition 2.12 and 2.13 of \cite{Mossel} and 
using Lemma 2.4 of \cite{Hastad12} we have the following
lemma. 
\begin{lemma}\label{lem-Efron-damp} Let 
 $(\Omega^{(1)}_i\times \Omega^{(2)}_i, \mu_i)_{i=1}^n$ be a sequence
of correlated spaces with $\rho_i = \rho(\Omega^{(1)}_i, \Omega^{(2)}_i;
\mu_i)$. Let $f : \prod_{i=1}^n\Omega^{(1)}_i \mapsto \R$ and 
$g : \prod_{i=1}^n\Omega^{(2)}_i \mapsto \R$, and let $g = \sum_{S}
g_S$ be the Efron-Stein decomposition of $g$ (refer to \cite{Mossel}
for a definition). Then,
$$ \E[f(x)g_S(y)] \leq \|f\|_2\|g\|_2\prod_{i\in S}\rho_i.$$
If the Efron-Stein decomposition of $g$ contains only functions of
weight at
least $s$ and $\rho = \max_i\rho_i$, then,
$$ \E[f(x)g(y)] \leq \rho^s\|f\|_2\|g\|_2.$$
The above also implies for the Bonami-Beckner operator $T_\rho$ that,
$$\|T_\rho f\|_2 \leq \rho^s\|f\|_2,$$
if the Efron-Stein decomposition of $f$ contains functions of weight
at least $s$.
\end{lemma}
The influence of a function on a product space is defined as follows.
\begin{definition}
Let $f$ be a measurable function on $(\prod_{i=1}^n\Omega_i,
\prod_{i=1}^n\mu_i)$. The influence of the $i$th coordinate on $f$ is:
$$\tn{Inf}_i(f) = \displaystyle \E_{\{x_j | j \neq
i\}}\left[\tn{Var}_{x_i}\left[f(x_1, x_2, \dots, x_i, \dots,
x_n)\right]\right].$$
\end{definition}
In particular, if $f : \{-1,1\}^n \mapsto \R$, and $f =
\sum_{\alpha\subseteq [n]}\wh{f}_\alpha\chi_\alpha$ is its Fourier
decomposition, then $\tn{Inf}_i(f) = \sum_{\alpha: i \in
\alpha}\wh{f}_\alpha^2$.

The following key results in Mossel's work~\cite{Mossel} shall be used
in the analysis of our reduction. We first restate Lemma 6.2 of
\cite{Mossel}. 
\begin{lemma} \label{lem-Mossel-big} Let $(\Omega_1^{(j)}, \dots,
\Omega_n^{(j)})_{j=1}^k$ be $k$ collections of finite probability
spaces such that $\{\prod_{j=1}^k\Omega^{(j)}_i \mid i=1,\dots, n\}$ are
independent. Suppose further that it holds for all
$i = 1, \dots, n$ that
$\rho(\Omega^{(j)}_i : 1\leq j\leq k) \leq \rho$. Then there exists an
absolute constant $C$ such that for,
$$\gamma = C\frac{(1- \rho)\nu}{\log \left(\slfrac{1}{\nu}\right)},$$
and $k$ functions $\left\{f_j \in L^2(\prod_{i=1}^n\Omega^{(j)}_i)
\right\}_{j=1}^k$,  the following holds,
$$\displaystyle \left|\E\left[\prod_{j=1}^kf_j\right] - 
\E\left[\prod_{j=1}^kT_{1-\gamma}f_j\right]\right| \leq
\nu\sum_{j=1}^k\sqrt{\tn{Var}[f_j]}\sqrt{\tn{Var}\left[\prod_{j'<j}
T_{1-\gamma} f_{j'}\prod_{j'>j}f_{j'}\right]}.$$ 
\end{lemma}
Our analysis shall also utilize the following bi-linear Gaussian
stability bound from \cite{Mossel} to locate influential coordinates.
\begin{theorem}\label{thm-Mossel-bilinear}
Let $(\Omega^{(1)}_i\times \Omega^{(2)}_i, \mu_i)$ be a sequence of
correlated spaces such that for each $i$, the probability of any atom
in $(\Omega^{(1)}_i\times \Omega^{(2)}_i, \mu_i)$ is at least $\alpha
\leq \slfrac{1}{2}$ and such that $\rho(\Omega^{(1)}_i,
\Omega^{(2)}_i; \mu_i) \leq \rho$ for all $i$. Then there exists a
universal constant $C$ such that, for every $\nu > 0$, taking 
$$ \tau =
\tn{exp}\left(C\frac{\log(\slfrac{1}{\alpha}) 
\log(\slfrac{1}{\nu})}{\nu(1-\rho)}\right),$$
for functions $f: \prod_{i=1}^n\Omega^{(1)}_i\mapsto [0,1]$ and $
g: \prod_{i=1}^n\Omega^{(2)}_i\mapsto [0,1]$ that satisfy,
$$\max \min_i(\tn{Inf}_i(f), \tn{Inf}_i(g)) \leq \tau,$$
for all $i$, we have,
$$\underline{\Gamma}_\rho(\E[f], \E[g]) - \nu \leq \E[fg] \leq
\ol{\Gamma}_\rho(\E[f], \E[g]) + \nu.$$
\end{theorem} 
 
Before describing the hardness reduction we define the following
useful distribution and state its properties.

\subsubsection*{Distribution $\mc{D}_{\delta, r}$}
We define the probability
measure $\mc{D}_{\delta,r}$ over the random variables $(X, 
Y = \{Y_i\}_{i=1}^r, Z = \{Z_i\}_{i=1}^r)$, where $X, Y_i,
Z_i \in \{-1,1\}$. A tuple $(X, Y, Z)$
is sampled from $\mc{D}_{\delta, r}$  by first choosing $X,
Y_1, \dots, Y_r \in \{-1,1\}$ independently and uniformly at random,
and setting each $Z_i = -Y_i$. Finally, with probability $\delta$, $j
\in [r]$ is chosen u.a.r and $Y_j$ and $Z_j$ are both set to $-X$. Let
 $X$, $Y$ and $Z$ define the correlated probability
spaces $\Omega^{(1)}$, $\Omega^{(2)}$
and $\Omega^{(3)}$ respectively with the joint probability measure
$\mc{D}_{\delta, r}$. Note that the marginal probability spaces
$(\Omega^{(2)}, \mc{D}_{\delta, r})$ and $(\Omega^{(3)}, \mc{D}_{\delta, r})$
are identical. Also, for $i \neq j \in [r]$, $Y_i$ is independent of
$Y_j$ and $Z_j$. 
It is easy to see the following lemma. 
\begin{lemma}\label{lem-correlation} For any probability $\delta$ and
integer $r > 0$,\\ 
(i) The minimum probability of an atom in $\mc{D}_{\delta, r}$ is at
least $\xi := \frac{\delta}{r2^r}$. \\
(ii) $\rho(\Omega^{(1)}, \Omega^{(2)}\times
\Omega^{(3)}; \mc{D}_{\delta, r}) \leq
\delta$. \\
 (iii) $\rho(\Omega^{(1)} \times \Omega^{(2)}, 
\Omega^{(3)}; \mc{D}_{\delta, r}) \leq 1 - \slfrac{\xi^2}{2} = 
\frac{\delta^2}{r^22^{2r+1}}$. \\ (iv) $\rho(\Omega^{(2)},
\Omega^{(3)}; \mc{D}_{\delta, r}) \leq 
1 - \xi^2/2$. \\
(v)  $\rho(\Omega^{(1)}, \Omega^{(2)}, 
\Omega^{(3)}; \mc{D}_{\delta, r}) \leq 1 - \slfrac{\xi^2}{2}$.
\end{lemma}
\begin{proof}
The first part can be computed by observing that the atom in
$\D_{\delta, r}$ with minimum probability is the one in which there is
a $j \in [r]$ such that $Y_j = Z_j$, and this atom has probability $\xi$ as
defined.
The second part is immediate since $X$ is independent of $(Y, Z)$ with
probability $1-\delta$. The third and fourth parts follow from (i) and
by showing that
Lemma \ref{lem-corbd} is applicable, which can be inferred in a manner
similar to the proof of connectedness in \cite{OW}. We omit the
details here. 
The fifth part follows from
Definition \ref{def-corr}.
\end{proof}

In the rest of this section we shall sometimes 
omit writing the joint distribution
along with $\Omega^{(1)}, \Omega^{(2)}$ and $\Omega^{(3)}$, as it will be clear
from the context.

\subsection{Hardness Reduction}\label{sec-redn-2color}
We begin with an instance $\Phi$ from Theorem \ref{thm-dto1multi} with 
the number of layers $L = \lceil
32\eps^{-2}\rceil$, for a parameter $\eps > 0$ which denotes the size
of the independent set in the NO Case.

\subsubsection{Construction of $G(H, E)$}
We continue with the construction of the instance $G(H, E)$, a 
$3$-uniform hypergraph. The construction uses a parameter $\delta$
which we shall fix later.

\medskip
\noindent
{\bf Vertices.} Consider a variable $v$ of $\Phi$ in layer $l$. Let
$H^v$ be a copy of $\{-1, 1\}^{R_l}$. The vertex set $H := \cup_{l\in
[L]}\cup_{v \in V_l}H^v$. The weight of a vertex $x \in H^v$ for $v
\in V_l$ is $2^{-R_l}/(L|V_l|)$. Thus, the total weight of all the
vertices corresponding to a particular layer is $1/L$.

\medskip
\noindent
{\bf Hyperedges.} Consider two variables $v \in V_l$ and $u \in
V_{l'}$ with a constraint $\pi_{v\rightarrow u}$ between them. Note
that for every $i \in R_{l'}$, $\left|\pi_{v\rightarrow
u}^{-1}(i)\right| = d^{l-l'}$. For convenience, we let $r = d^{l-l'}$,
and dropping the subscript we shall refer to the projection simply as $\pi$.  
Let $x \in H^{u}$ and $y, z \in H^{v}$ be chosen by sampling
 $(x_i, y|_{\pi^{-1}(i)},
z|_{\pi^{-1}(i)})$ from $(\Omega^{(1)}\times \Omega^{(2)}\times \Omega^{(3)};
\mc{D}_{\delta, r})$ independently for each 
$i \in R_{l'}$. Let $\mc{D}^{vu}$ denote 
the probability distribution of the choice of $(x, y,
z)$. For all such
$(x, y,z)$ in the support of $\mc{D}^{vu}$ 
add a hyperedge between these three vertices $x, y$ and $z$.

\subsection{YES Case}
In the YES Case, let $\sigma$ be the labeling to the variables that
satisfies all constraints in $\Phi$. For every vertex $x \in H^v$
for a variable $v$ in layer $l$, color $x$ with $x_{\sigma(v)}$. It is
easy to see from the above construction of the hyperedges that this is
a valid 
$2$-coloring of the hypergraph.

\subsection{NO Case}
Suppose that there is an independent set of $\eps > 0$ fraction of
vertices. For a variable $v$ of $\Phi$, let $f_v$ be the indicator of the
independent set in the long code $H^v$. Let the \emph{heavy} variables
$v$ be such that $\E[f_v] \geq \frac{\eps}{2}$. 
After averaging and arguments analogous to those in Section
\ref{sec-multi-No} we obtain two layers $l <
l'$ such that the \emph{heavy} variables in these two layers induce
at least $\frac{\eps^2}{64}$ fraction of constraints between these two layers.
As before, we set $r = d^{l-l'}$. Also, we shall denote $R_{l'}$
by $R_1$ and $R_{l}$ by $R_2$.

We need to show that,
\begin{equation}
\E_{v,u}\left[\E_{(x,y,z)\leftarrow
\mc{D}^{vu}}\left[f_u(x)f_v(y)f_v(z)\right]\right] > 0,
\label{eqn-toshow}
\end{equation}
where the outer expectation is over pairs of heavy variables $v \in
V_l$ and $u \in V_{l'}$ which share a constraint. The analysis
consists of two main steps. In the first step we show that unless
$f_u$ and $f_v$ share influential coordinates, one can re-randomize
the $x$ variable to be independent in the inner expectation of 
Equation \eqref{eqn-toshow}. However, the notion of influence of $f_v$ 
used in
this step depends on the choice of $u$. 

The
second step shows that for a non-trivial fraction of heavy 
neighbors $u$ of $v$, the notion of
influence used in the first step can be made independent of $u$. In
addition it shows that for these $u$, the marginal expectation 
$\E[f_v(y)f_v(z)]$
induced by $\mc{D}^{vu}$ is bounded away from zero. This step
crucially uses the smoothness property of the PCP.  

\subsubsection{Making $x$ independent} 
Let us fix a pair of heavy
vertices $v, u$ which share a constraint $\pi$. For convenience we
shall think of the distribution $\mc{D}^{vu}$ being on
$\otimes_{i \in R_1}(x_i, y|_{\pi^{-1}(i)}, z|_{\pi^{-1}(i)})$, where
each $(x_i, y|_{\pi^{-1}(i)}, z|_{\pi^{-1}(i)})$ is sampled
independently from $(\Omega^{(1)}\times \Omega^{(2)}\times \Omega^{(3)},
D_{\delta, r})$. We represent the space of
 $(x_i, y|_{\pi^{-1}(i)}, z|_{\pi^{-1}(i)})$ by the correlated space 
$(\Omega^{(1)}_i \times \Omega^{(2)}_i\times \Omega^{(3)}_i)$, which
is
an independent copy of $(\Omega^{(1)}\times \Omega^{(2)}\times \Omega^{(3)})$. Thus,
the space of $\otimes_{i\in R_1}(x_i, y|_{\pi^{-1}(i)},
z|_{\pi^{-1}(i)})$ is 
$\prod_{i \in R_1} 
(\Omega^{(1)}_i \times \Omega^{(2)}_i\times \Omega^{(3)}_i)$. 
The $i$th coordinate
influence of a function $f$ on $\prod_{i \in R_1}\Omega^{(2)}_i = 
\prod_{i \in R_1}\Omega^{(3)}_i$ is
denoted by $\ol{\tn{Inf}}_i(f_v)$. The probability 
measure on all
these spaces is induced by $\mc{D}^{vu}$. 

Using the above and since the functions $f_u$ and $f_v$ are all in
the range $[0,1]$ we have the following lemma which follows from Lemma
\ref{lem-correlation} and Lemma \ref{lem-Mossel-big}. 
\begin{lemma}\label{lem-addnoise}
There is a universal
constant $C$ such that for an arbitrarily small choice of $\nu > 0$,
letting $\gamma = C\frac{\nu\xi^2}{2\log(1/\nu)}$, the following
holds,
\begin{equation}
\left| E[f_u(x)f_v(y)f_v(z)] -
E[T_{1-\gamma}f_u(x)\ol{T}_{1-\gamma}f_v(y)\ol{T}_{1-\gamma}f_v(z)]\right| 
\leq
\nu,
\end{equation}
\begin{equation}
\left| E[f_v(y)f_v(z)] -
E[\ol{T}_{1-\gamma}f_v(y)\ol{T}_{1-\gamma}f_v(z)]\right| \leq
\nu,
\end{equation}
where the $T_{1-\gamma}$ is the Bonami-Beckner operator over
$\{-1,1\}^{R_1} = \prod_{i\in R_1}\Omega^{(1)}$ and 
$\ol{T}_{1-\gamma}$ is the Bonami-Beckner operator over the space 
$\prod_{i\in R_1}\Omega^{(2)}_i = \prod_{ \in R_1}\Omega^{(3)}_i$.
To be precise, $\ol{T}_{1-\gamma}$ resamples from each
$\Omega^{(2)}_i$ independently with probability $\gamma$. Note that
$\ol{T}_{1-\gamma}$ depends on the constraint 
$\pi$ and hence on the choice of $u$. 
\end{lemma}
Using a value of $\gamma$ which we shall obtain from the above lemma,
consider the function $ F(y,z) = 
\ol{T}_{1-\gamma}f_v(y)\ol{T}_{1-\gamma}f_v(z)$ 
over the
space $\prod_{i \in R_1}(\Omega^{(2)}_i\times \Omega^{(3)}_i)$. For
the time being let $f'$ denote $\ol{T}_{1-\gamma}f_v$ and $f'_i$ denote the
function $f'$ depending only on the $i$th space $\Omega^{(2)}_i = 
\Omega^{(3)}_i$ where
 the fixing of the rest of the coordinates will be clear from the
context. Thus, $F(y,z) = f'(y)f'(z)$. The
$i$th influence of $F$ in the space 
 $\prod_{i \in R_1}(\Omega^{(2)}_i\times \Omega^{(3)}_i)$
can be written as:
\begin{eqnarray}
\displaystyle \ol{\tn{Inf}}_i(F)
& = & \frac{1}{2}\E_{\substack{(y|_{\pi^{-1}(j)}, z|_{\pi^{-1}(j)})\leftarrow 
(\Omega^{(2)}_j\times \Omega^{(3)}_j) \\ j \in R_1\setminus\{i\}}}
\Big[   \nonumber \\
& &   \E_{((Y_1, Z_1), (Y_2, Z_2))\leftarrow (\Omega^{(2)}\times
\Omega^{(3)})^2}\left[(f_i'(Y_1)f_i'(Z_1) -
f_i'(Y_2)f_i'(Z_2))^2\right]\Big] \label{eqn-infprod}
\end{eqnarray}
The following inequality was proved in Lemma 4 of the work of
Samorodnitsky and Trevisan~\cite{ST-Gowers}.
\begin{lemma}\label{lem-aibi}
Let $a_1, a_2, b_2, b_2 \in [-1,1]$. Then, $(a_1a_2 - b_1b_2)^2 \leq
2\left((a_1 - b_1)^2 + (a_2 - b_2)^2\right)$. 
\end{lemma}
Using the above lemma we obtain the following bound.
\begin{lemma}\label{lem-prodinfbd}
From the definitions used above,
$$\displaystyle \ol{\tn{Inf}}_i(F) 
\leq 4\ol{\tn{Inf}}_i(f').$$
\end{lemma}
\begin{proof}
Using Lemma \ref{lem-aibi} and the fact that $f'$ is bounded in
$[0,1]$ we can upper bound $\ol{\tn{Inf}}_i(F)$ in Equation
\eqref{eqn-infprod} by
\begin{eqnarray}
& & \frac{1}{2}\E_{\substack{(y|_{\pi^{-1}(j)}, z|_{\pi^{-1}(j)})\leftarrow 
(\Omega^{(2)}_j\times \Omega^{(3)}_j) \\ j \in R_1\setminus\{i\}}}
\Big[ \nonumber \\
&& \E_{((Y_1, Z_1), (Y_2, Z_2))\leftarrow (\Omega^{(2)}\times
\Omega^{(3)})^2}\left[(f_i'(Y_1) - f_i'(Y_2))^2 +
(f_i'(Z_1) - f_i'(Z_2))^2\right]\Big] \nonumber \\
& = & 4\ol{\tn{Inf}}_i(f').\nonumber
\end{eqnarray}
\end{proof}
We also have the following lemma.
\begin{lemma}\label{lem-blocktobinary}
Let $\tn{Inf}_j$ be the $j$th coordinate influence over
the space $\{-1,1\}^{R_2}$ equipped with the uniform measure. 
Then, for $i \in R_1$, 
$\ol{\tn{Inf}}_i(f') \leq r\sum_{j\in \pi^{-1}(i)}
\tn{Inf}_j(f')$.
\end{lemma}
\begin{proof} By the definition of influence, the LHS of the 
assertion can be
written as,
\begin{equation}
\frac{1}{2}\E_{\substack{y|_{\pi^{-1}(j)}\leftarrow 
\Omega^{(2)}_j \\ j \in R_1\setminus\{i\}}}
\left[
\E_{(Y^0, Y^r)\leftarrow (\Omega^{(2)})^2}\left[(f'_i(Y^0) -
f'_i(Y^r))^2\right]\right].
\end{equation}
Order the coordinates in $\pi^{-1}(i)$ as $1,\dots, r$ and define
(depending on the choice of $Y^0$ and $Y^r$) a sequence $Y^1, \dots,
Y^{r-1}$ where $Y^k$ contains the value of the first $r - k$
coordinates from $Y^0$ and the rest from $Y^r$. Letting $R_2' :=
R_2\setminus \pi^{-1}(i)$, the above
expression can be rewritten as,
\begin{eqnarray}
& & \frac{1}{2}\E_{y|_{R_2'} \leftarrow 
\{-1,1\}^{R_2'}}\left[\E_{(Y^0, Y^r)\leftarrow
(\{-1,1\}^r)^2}\left[\left(\sum_{k=0}^{r-1}(f_i'(Y^k) -
f_i'(Y^{k+1})\right)^2\right]\right] \nonumber \\
& \leq & \frac{1}{2}\E_{y|_{R_2'} \leftarrow 
\{-1,1\}^{R_2'}}\left[\E_{(Y^0, Y^r)\leftarrow
(\{-1,1\}^r)^2}\left[r\sum_{k=0}^{r-1}(f_i'(Y^k) -
f_i'(Y^{k+1})^2\right]\right] \nonumber \\
& = & r\sum_{j\in \pi^{-1}(i)}
\tn{Inf}_j(f'),
\end{eqnarray}
where we used Cauchy-Schwarz to obtain the first inequality.
\end{proof}
The following lemma uses the above analysis to show
that $x$ can be made independent of $y$ and
$z$ without incurring much loss, unless $f_u$ and $f_v$ have
matching influential coordinates.
\begin{lemma}\label{lem-main-inf} 
There is a universal constant $C$ such that 
for an arbitrarily small constant $\nu > 0$, and 
$$\gamma =
 \frac{\nu\xi^2}{2\log(1/\nu)} \ \ , \ \  \tau =
\nu^{C\frac{\log(1/\xi)\log(1/\nu)}{\nu(1-\delta)}},$$
unless there is $i \in R_1$ such that,
\begin{equation}
\min(\tn{Inf}_i(T_{1-\gamma}f_u), 4r\sum_{j\in \pi^{-1}(i)}
\tn{Inf}_j(\ol{T}_{1-\gamma}f_v)) \geq \tau, \label{eqn-inf-cond}
\end{equation}
we have,
$$\E[f_u(x)f_v(y)f_v(z)] \geq 
\underline{\Gamma}_{\delta}(\E[f_u],
\E[f_v(y)f_v(z)] - \nu) - 2\nu.$$
\end{lemma}
\begin{proof}
Suppose that there exists no $i \in R_1$ as in the condition of the
lemma. Using Lemmas \ref{lem-prodinfbd} and \ref{lem-blocktobinary}
our supposition implies that there exists no $i \in R_1$ such that,
$$\min(\tn{Inf}_i(T_{1-\gamma}f_u),
\ol{\tn{Inf}}_i(F)) \geq
\tau,$$
where $F(y,z)$ was defined as
$\ol{T}_{1-\gamma}f_v(y)\ol{T}_{1-\gamma}f_v(z)$.
Using Theorem \ref{thm-Mossel-bilinear} and 
Lemma \ref{lem-correlation} the above implies,
\begin{eqnarray}
\displaystyle \E[T_{1-\gamma}f_u(x)\ol{T}_{1-\gamma}f_v(y)\ol{T}_{1-\gamma}f_v(z)]
& \geq & \underline{\Gamma}_{\delta}
\left(\E[T_{1-\gamma}f_u(x)], 
\E[\ol{T}_{1-\gamma}f_v(y)\ol{T}_{1-\gamma}f_v(z)]\right) - \nu
\nonumber \\
& = &   \underline{\Gamma}_{\delta}
\left(\E[f_u(x)], 
\E[\ol{T}_{1-\gamma}f_v(y)\ol{T}_{1-\gamma}f_v(z)]\right) - \nu. 
\end{eqnarray}
Using Lemma \ref{lem-addnoise} the above implies that
\begin{equation}
\E[f_u(x)f_v(y)f_v(z)] \geq 
\underline{\Gamma}_{\delta}\left(\E[f_u(x)], 
\E[f_v(y)f_v(z)] - \nu \right) - 2\nu.
\end{equation}
\end{proof}
Note that there are two issues that are left to resolve. Firstly, we
need to lower bound  $\E[f_v(y)f_v(z)]$. Secondly,
$\tn{Inf}_i(\ol{T}_{1-\gamma}(f_v))$ depends on the choice of $u$. We
shall identify a significant fraction of heavy neighbors $u$ of $v$, for
which the expectation is bounded as well as 
$\tn{Inf}_i(\ol{T}_{1-\gamma}(f_v)) \approx 
\tn{Inf}_i(T_{1-\gamma}(f_v))$, the latter being independent of $u$.
For this we shall utilize the smoothness property of the PCP.

\subsubsection{Identifying \emph{good} neighbors $u$}
Let us first set a parameter $s$ as, 
$$s := \max\left(\frac{r}{\xi}\ln\left(\frac{1}{\nu}\right), 
\frac{r}{2\gamma}\ln\left(\frac{32r^2}{\tau}\right)\right).$$
Let the Efron-Stein decomposition of $f_v$ with respect to
$\{-1,1\}^{R_2}$ be,
\begin{equation}
f_v = \sum_{\alpha \subseteq R_2}\wh{f}_{v,\alpha}\chi_\alpha.
\label{eqn-Efron1}
\end{equation}
It can be seen (see \cite{Hastad12}) that the Efron-Stein
decomposition of $f_v$ with respect to $\prod_{i\in
R_1}\Omega^{(2)}_i$ is,
\begin{equation}
f_v = \sum_{\beta \subseteq R_1}f_v^\beta, \label{eqn-Efron2}
\end{equation}
where,
\begin{equation}
f_v^\beta = \sum_{\substack{\alpha \subseteq R_2 \\ \pi(\alpha) =
\beta}} \wh{f}_{v, \alpha}\chi_\alpha \label{eqn-Efron3}
\end{equation}
We say that a subset $\alpha$ is \emph{shattered} by $\pi =
\pi_{v\rightarrow u}$ if
$|\pi(\alpha)| = |\alpha|$. Using this we decompose 
$f_v$ into three functions, depending on the choice of $u$, as follows
\begin{eqnarray}
f_1 & = & \sum_{\alpha : |\alpha|\geq s}\wh{f}_{v,\alpha}\chi_\alpha
 \label{eqn-f1} \\
f_2 & = &  \sum_{\substack{\alpha : |\alpha| < s \\ \alpha \tn{ not
shattered}}}\wh{f}_{v,\alpha}\chi_\alpha
 \label{eqn-f2} \\
f_3 & = &  \sum_{\substack{\alpha : |\alpha| < s \\ \alpha \tn{
shattered}}}\wh{f}_{v,\alpha}\chi_\alpha
\label{eqn-f3}
\end{eqnarray}
To identify the \emph{good} neighbors of $v$, we need the following
key lemma. 
\begin{lemma}\label{lem-shatter}
With expectation taken over a random neighbor $u \in V_{l'}$ which
shares a constraint with $v$, $\E[\|f_2\|_2] \leq (s/\sqrt{T})$. 
Here $T$ is the smoothness parameter from Theorem \ref{thm-dto1multi}. 
\end{lemma}
\begin{proof} For a given $\alpha \subseteq R_2$ such that $|\alpha|
< s$, the probability (over $u$) that it is not shattered is at most
$$\sum_{i\neq j \in \alpha} \Pr[\pi_{v\rightarrow u}(i) =
\pi_{v\rightarrow u}(j)] \leq \frac{s^2}{T}.$$
Since, $\sum \wh{f}_{v, \alpha}^2 \leq 1$, we obtain that,
$$\E[\|f_2\|_2] \leq \left(\E[\|f_2\|_2^2]\right)^{1/2} \leq
\frac{s}{\sqrt{T}}.$$
\end{proof}
The above lemma implies that for at least $1 - (s^2/T)^{1/4}$ fraction
of the neighbors $u \in V_{l'}$ of $v$, $\|f_2\|_2 \leq
(s^2/T)^{1/4}$. Call such neighbors
$u$ of $v$ which satisfy this bound as \emph{good}.

\subsubsection*{Lower bounding $\E[f_v(y)f_v(z)]$}
We first set $\eta = \frac{2\delta}{r}$. It is easy to see that for
any $j \in R_2$, $\E[y_jz_j] = -1\left(1-\frac{\eta}{2}\right) +
\frac{\eta}{2} = -1 + \eta$. We shall first lower bound
$\E[f_v(y)T_{1 - \eta}f_v(-y)]$. 
We shall need the following lemma from \cite{MORSS06} which is obtained
using the \emph{reverse} hypercontractive inequality over the
boolean domain. 
\begin{lemma}\label{lem-Hatami}
Let $A, B \subseteq \{-1, 1\}^n$ have relative densities,
$$\frac{|A|}{2^n} = e^{-a^2/2} \ \ \ \ \ \ \ \ \ \ \frac{|B|}{2^n} =
e^{-b^2/2},$$
and let $y \in \{-1,1\}$ be uniform and $y'$ be a $\rho$-correlated
copy of $y$, i.e. $\E[y_iy'_i] = \rho,$ independently for each $i \in [n]$, for some $\rho > 0$. 
Then,
\begin{equation}
\Pr[y \in A, y' \in B] \geq exp\left[-\frac{1}{2}\cdot \frac{a^2 + b^2
+ 2\rho ab}{1 - \rho^2}\right].
\end{equation}
\end{lemma}
Since $f_v$ is an indicator function let $A = \{y \mid f_v(y) = 1\}$. As
$v$ was chosen to be heavy, we have $\E[f_v] \geq  \frac{\eps}{2}$.
Let $B = -A$, i.e. $B = \{-y \mid y \in A\}$. It is easy to see that
\begin{equation}
\E[f_v(y)T_{1 - \eta}f_v(-y)] = \Pr[y \in A, y' \in B],
\end{equation}
where $y'$ is a $1-\eta$ correlated copy of $y$. Using Lemma
\ref{lem-Hatami} we obtain,
\begin{equation}
\E[f_v(y)T_{-1 + \eta}f_v(y)] \geq
\left(\frac{\eps}{2}\right)^{4/\eta}. \label{eqn-folk}
\end{equation}

The following two lemmas decompose two expectations we are
interested in.
\begin{lemma}\label{lem-decomp-1} Using the decompositions above, 
\begin{equation}
\left|\E[f_v(y)T_{1 - \eta}f_v(-y)] - \E[f_3(y)T_{1 -
\eta}f_3(-y)]\right| \leq 
2\|f_2\|_2 + 2\nu.
\end{equation}
\end{lemma}
\begin{proof}
By Lemma \ref{lem-Efron-damp} and Equation \eqref{eqn-Efron1}, we have 
$$|\E[f_v(y)T_{1 - \eta}f_1(-y)]| \leq \|f_v\|_2\|f_1\|_2(1-\eta)^s \leq
\nu,$$
by our setting of $s$ and since $\|f_v\|_2, \|f_1\|_2 \leq 1$.
Furthermore,
$$|\E[f_v(y)T_{1 - \eta}f_2(-y)]| \leq  \|f_v\|_2\|f_2\|_2 \leq
\|f_2\|_2.$$
We can repeat the above with $\E[f_v(y)T_{1 - \eta}f_3(-y)]$ using the
fact that $\|T_{1 - \eta}f_3(-y)\|_2 \leq 1$ to obtain the lemma.
\end{proof}
\begin{lemma}\label{lem-decomp-2} Using the decompositions above and
having $(y, z)$ sampled from $(\prod_{i \in R_1}(\Omega^{(2)}_i\times
\Omega^{(3)}_i); \mc{D}^{vu})$, 
\begin{equation}
\left|\E[f_v(y)f_v(z)] - \E[f_3(y)f_3(z)]\right| \leq 
2\|f_2\|_2 + 2\nu.
\end{equation}
\end{lemma}
\begin{proof}Using the bound (iv) of Lemma \ref{lem-correlation}, the
decomposition in Equations \eqref{eqn-Efron2} and \eqref{eqn-Efron3},
and Lemma \ref{lem-Efron-damp} we obtain,
$$|\E[f_v(y)f_1(z)]| \leq \|f_v\|_2\|f_1\|_2(1-\eta)^{s/r} \leq
\nu,$$
by our setting of $s$. The rest of the proof is analogous to Lemma
\ref{lem-decomp-1}.
\end{proof}
Note that $y_i$ is independent of $y_j$ and $z_j$ for $i \neq j \in R_2$.
Also,  when sampling $z$ given
$y$ the coordinates in a shattered
subset $\alpha$ are flipped independently with probability $1 -
\frac{\eta}{2}$. Thus,
$$ \E[f_3(y)f_3(z)] = \sum_{\substack{\alpha : |\alpha| < s \\ \alpha \tn{
shattered}}}\wh{f}_{v,\alpha}^2(-1 + \eta)^{|\alpha|} = \E[f_3(y)T_{1
-
\eta}f_3(-y)].$$
From the above
analysis, Lemma \ref{lem-shatter}, and Equation \eqref{eqn-folk}, 
we have that for all good neighbors $u$ of  $v$, 
\begin{equation}
\E[f_v(y)f_v(z)] \geq \left(\frac{\eps}{2}\right)^{4/\eta} 
- 4\left(\frac{s^2}{T}\right)^{1/4}
- 4\nu, 
\end{equation}
where $y$ and $z$ are sampled according to $\mc{D}^{vu}$.

\subsubsection*{Showing $\tn{Inf}_i(\ol{T}_{1-\gamma}f_v) \approx 
\tn{Inf}_i(T_{1-\gamma}f_v)$}
Recall that $\ol{T}_{1-\gamma}$ is the Bonami-Beckner operator on the
space $\prod_{i \in R_1}\Omega^{(2)}_i$ and $T_{1-\gamma}$ is over
$\{-1,1\}^{R_2}$ equipped with the uniform measure. Let $h =
\ol{T}_{1-\gamma}f_v$ and $g = T_{1-\gamma}f_v$. Define the functions
$h_i := \ol{T}_{1-\gamma}f_i$ and $g_i := T_{1-\gamma}f_i$ for
$i=1,2,3$.

Since the operators $\ol{T}_{1-\gamma}$ and $T_{1-\gamma}$ are
contractions, by Lemma \ref{lem-shatter} we
have that for good neighbors $u$, $\|h_2\|_2, \|g_2\|_2  \leq
(s^2/T)^{1/4}$. 
Also, by Lemma \ref{lem-Efron-damp} and Efron-Stein decompositions of 
$f_v$ (Equations \eqref{eqn-Efron1}, \eqref{eqn-Efron2} and
\eqref{eqn-Efron3}), we obtain: $\|h_1\|_2 \leq
(1- \gamma)^{s/r}$ and $\|g_1\|_2 \leq (1- \gamma)^s$. By our setting of
$s$, we get $\|h_1\|_2^2, \|g_1\|_2^2 \leq \frac{\tau}{32r^2}$. 

For a subset $\alpha$ which is shattered, it is easy to see that
that $\wh{h}_\alpha = \wh{g}_\alpha = \wh{f}_{v, \alpha}(1 -
\gamma)^{|\alpha|}$. Using the definition of influence over the domain
$\{-1,1\}^{R_2}$ we obtain the following lemma.
\begin{lemma}\label{lem-infsame}
For any $i \in R_2$, 
$$\left|\tn{Inf}_i(\ol{T}_{1-\gamma}f_v) -
\tn{Inf}_i(T_{1-\gamma}f_v)\right| \leq
2\left(\frac{s^2}{T}\right)^{1/4} + \frac{\tau}{16r^2}.$$
\end{lemma}

\medskip
\noindent
{\bf Choice of Parameters.} Given $\eps > 0$, fix  $\delta \in (0,
1/2)$, which also fixes $\eta$. The choice of $L$ made at the
beginning of Section \ref{sec-redn-2color} is fixed
and therefore the maximum possible value of $r$ is also fixed. 
Choose $\nu$ small enough so that 
\begin{equation}
\underline{\Gamma}_{\delta}
\left(\frac{\eps}{2}, 
 \frac{1}{2}\left(\frac{\eps}{2}\right)^{4/\eta} - 5\nu \right) - 2\nu
> 0.\label{eqn-nonzero}
\end{equation} 
This also fixes the choice of $\gamma$ and $\tau$ by Lemma
\ref{lem-main-inf}, and the choice of $s$ as defined above. 
Then choose $T$ to be large enough so that 
$$4\left(\frac{s^2}{T}\right)^{1/4} \leq  
\min\left\{\frac{1}{2}\left(\frac{\eps}{2}\right)^{4/\eta}, 
\frac{\eps^2}{128}\right\},$$
and,
$$2\left(\frac{s^2}{T}\right)^{1/4} \leq  \frac{\tau}{16r^2}.$$
The above setting implies that
for all good neighbors $u$ of  $v$, 
\begin{equation}
\E[f_v(y)f_v(z)] \geq \frac{1}{2}\left(\frac{\eps}{2}\right)^{4/\eta} 
- 4\nu, \label{eqn-expec}
\end{equation}
and for any $i \in R_2$, using Lemma \ref{lem-infsame},
\begin{equation}
\left|\tn{Inf}_i(\ol{T}_{1-\gamma}f_v) -
\tn{Inf}_i(T_{1-\gamma}f_v)\right| \leq
\frac{\tau}{8r^2}.\label{eqn-inf-same2}
\end{equation}
Using Equations \eqref{eqn-nonzero}, \eqref{eqn-expec} and
\eqref{eqn-inf-same2} along with Lemma 
\ref{lem-main-inf} for a heavy \emph{and} good neighbor $u$ of $v$
yields an $i^* \in R_1$ such that,
\begin{equation}
\min(\tn{Inf}_{i^*}(T_{1-\gamma}f_u), 4r\sum_{j\in \pi^{-1}(i^*)}
\tn{Inf}_j(T_{1-\gamma}f_v)) \geq \tau/2. \label{eqn-inf-cond2}
\end{equation}

\medskip
\noindent
{\bf Labeling.} The labeling to a heavy variable $u \in V_{l'}$ is
given by choosing a label $i \in R_1$ independently 
with probability proportional to $\tn{Inf}_i(T_{1-\gamma}f_u)$.  The label
 of a heavy variable $v \in V_{l}$ is similarly assigned
given by choosing $j \in R_2$ independently 
with probability proportional to $\tn{Inf}_i(T_{1-\gamma}f_v)$.
Note
that the sum of all influences of $T_{1-\gamma}f_u$
($T_{1-\gamma}f_v$) is bounded by $1/\gamma$.

Suppose $u$ is a good and heavy neighbor of a heavy variable $v$. 
Then analysis above along with
Lemma \ref{lem-main-inf} and Equation \ref{eqn-inf-cond2} 
implies that the labeling strategy will
succeed for $v$ and $u$ with probability $\tau^2\gamma^2/16r$.
Additionally, from the above analysis, at least $\frac{\eps^2}{128}$
fraction of constraints between layers $l$ and $l'$ are between heavy
variables $v \in V_l$ and $u \in V_{l'}$ such that $u$ is a good
neighbor of $v$. Thus, the probabilistic labeling strategy satisfies
in expectation $\frac{\eps^2\tau^2\gamma^2}{2048r}$ fraction of
constraints. By choosing the soundness $\zeta$ to be small enough we
obtain a contradiction. 

\bibliographystyle{alpha}
\bibliography{Refs-2color}
\appendix

\section{Construction of Smooth $d$-to-$1$ MLPCP}\label{sec-dto1multi}
The construction of the Smooth $d$-to-$1$ Multi-Layered PCP $\Phi$ closely 
follows the construction used \cite{Khot-3}. We shall only give the
construction. 

We begin with an instance $\mc{L}$ of the $d$-to-$1$ Game given by
Conjecture \ref{conj-dto1bireg} with the variable sets $\mc{U}$,
$\mc{V}$ and label sets $[k]$ and $[m]$. For convenience we refer to
the variables in $\mc{V}$ as $\mc{V}$-variables and those in $\mc{U}$ as
$\mc{U}$-variables. 

The variables  of $\Phi$ in
the $l$th layer are sets of $(TL + L - l)$ $\mc{V}$-variables
and $(l-1)$ $\mc{U}$-variables. 
The label set $R_l$ of layer $l$ is the set of all $(TL + L-1)$-tuples of 
labelings to  $TL + L - l$ $\mc{V}$-variables and $l-1$
$\mc{U}$-variables.

There is a constraint between a
variable $v$ in layer $l$ and a variable $u$ in layer $l'$ of $\Phi$
if replacing $(l - l')$ $\mc{V}$-variables $q_1,\dots, q_{l-l'}$ from
the set associated with $v$, with $\mc{U}$-variables 
$p_1, \dots, p_{l-l'}$ such that $p_r$ has a constraint with $q_r$ in
$\mc{L}$ for $r = 1, \dots, l-l'$, yields the set associated with $u$.
The constraint $\pi_{v\rightarrow u}$ is projection which  checks
the consistency of the labels, according to whether the variables of
$\mc{L}$ common to both $u$ and $v$ are assigned identically and the
assignments to $p_1, \dots, p_{l-l'}$ and $q_1,\dots, q_{l-l'}$ are
consistent. It is easy to see that $\pi_{v\rightarrow u}^{-1}(i) =
d^{l-l'}$ for any $i \in R_{l'}$.

The proof of weak density follows from the bi-regularity property of
$\Phi$ in a manner identical to the proof in \cite{DGKR03}. The proof of
soundness is identical to the proof in \cite{Khot-3}. The proof of
hardness in Theorem \ref{sec-dto1multidef} follows from standard
arguments as given in \cite{DGKR03}. We omit these proofs.

\section{Fourier Analysis}\label{sec-Fourier}
We will be working over the field $\Ft$.
Define the following homomorphism
$\phi$ from $(\Ft, +)$ to the multiplicative group $(\{-1,1\}, .)$, by
$\phi(a) := (-1)^a$.
We now consider the vector space $\Ft^m$ for some positive integer $m$. We
define the `characters' $\chi_\alpha : \Ft^m \mapsto \{-1,1\}$ for every
$\alpha \in \Ft^m$ as,
$$ \chi_\alpha(f) := \phi(\alpha\cdot f), \ \ \ \ \ \ \ f \in \Ft^m$$
where `$\cdot$' is the inner product in the vector space $\Ft^m$. 
The characters $\chi_\alpha$ satisfy the following properties,
\begin{eqnarray*}
\chi_0(f)  = 1 &\ \ \ \ \ & \forall f \in \Ft^m\\
\chi_\alpha(0) = 1 &\ \ \ \ \ & \forall \alpha \in \Ft^m\\
\chi_{\alpha + \beta}(f) = \chi_\alpha(f)\chi_\beta(f)\\
\chi_{\alpha}(f + g) = \chi_\alpha(f)\chi_\alpha(g)
\end{eqnarray*}
and,
\begin{equation*}
\E_{f \in \Ft^m}\left[\chi_\alpha(f)\right] = \begin{cases}
					1 & \text{ if } \alpha = 0 \\
				  	0   & \text{ otherwise }
					\end{cases}
\end{equation*}
The characters $\chi_\alpha$ form an orthonormal basis for $L^2(\Ft^m)$. We
have,
$$  \left<\chi_\alpha, \chi_\beta\right> = \begin{cases}
				1 & \text{ if }\alpha = \beta \\
				0 & \text{ otherwise }
				\end{cases}
$$
where,
$$ \left<\chi_\alpha, \chi_\beta\right> := \E_{f \in \Ft^m}\left[
\chi_\alpha(f)\chi_\beta(f)\right]. $$

Let $A:\Ft^m \mapsto \R$ be any real valued function. Then the
Fourier expansion of $A$ is given by,
$$A(x) = \sum_{\alpha \in \Ft^m}\widehat{A}_\alpha\chi_\alpha(x),$$
where,
$$\widehat{A}_\alpha = \E_{x\in \Ft^m}[A(x)\chi_\alpha(x)].$$
A useful equality is:
$$\widehat{A}_0 = \E_{x\in \Ft^m}[A(x)].$$

\subsection*{Folding}
The following lemma gives a property of the Fourier coefficients of
any homogeneously folded function.
\begin{lemma}
Let $A : \Ft^m\mapsto \R$ be any function such that $A(x + y)
= A(x)$ for some $y \in \Ft^m$ and all $x\in \Ft^m$. 
Then if $\widehat{A}_\alpha \neq 0$, then
$\alpha\cdot y = 0$.
\end{lemma}
\begin{proof} Assume $\widehat{A}_\alpha \neq 0$.
By definition and using the folding property,
\begin{eqnarray*}
\widehat{A}_\alpha & =  & \E_{x\in \Ft^m}[A(x)\chi_\alpha(x)] \\ 
                   & =  &\E_{x\in \Ft^m}[A(x + y)
			\chi_\alpha(x +  y)] \\
		   & = & \E_{x\in \Ft^m}[A(x)
			\chi_\alpha(x + y)] \\
		& = & \E_{x\in \Ft^m}[A(x)
			\chi_\alpha(x)]\chi_{\alpha}(y) \\
		& = & \widehat{A}_\alpha \chi_\alpha(y).
\end{eqnarray*}
Thus, if $\widehat{A}_\alpha \neq 0$, then $\chi_\alpha(y) = 1$.
Thus, $\phi(\alpha\cdot y) = 1$. 
This implies that $\alpha\cdot y = 0$.
\end{proof}

\end{document}